\documentclass[aps,showpacs,prl,twocolumn,superscriptaddress,longbibliography]{revtex4-1}
\usepackage{amsmath,amssymb,amsfonts,bm}
\usepackage{xcolor}
\usepackage{bbold}
\usepackage{graphicx}
\usepackage{dcolumn}
\usepackage[left = 1in, top = 1in, right = 1in, bottom = 1in]{geometry}
\usepackage{epstopdf}
\usepackage{epsfig}
\usepackage{url}
\usepackage{hyperref}
\usepackage{verbatim}
\usepackage[export]{adjustbox}
\usepackage{scalerel}
\usepackage{braket}
\usepackage[normalem]{ulem}
\usepackage{float}

\usepackage{mathtools}

\newcommand{\be}{\begin{equation}}
\newcommand{\ee}{\end{equation}}
\newcommand{\bs}{\begin{split}}
\newcommand{\es}{\end{split}}
\newcommand{\rvec}{\mathbf{r}}

\newcommand{\Qvec}{\mathbf{q}}
\newcommand{\xvec}{\mathbf{x}}
\newcommand{\evec}{\mathbf{e}}

\newcommand{\GS}{\mathrm{GS}}
\newcommand{\SMA}{\mathrm{SMA}}

\newcommand{\Hres}{H_{\mathrm{res}}}
\newcommand{\Hdiag}{H_{\mathrm{diag}}}
\newcommand{\dop}{\mathcal{D}}

\usepackage{mathrsfs}
\usepackage{bbold}
\usepackage{dcolumn}
\usepackage[left = 1in, top = 1in, right = 1in, bottom = 1in]{geometry}
\usepackage{epstopdf}
\usepackage{epsfig}
\usepackage{url}
\usepackage{hyperref}
\usepackage{verbatim}
\usepackage[export]{adjustbox}
\usepackage{scalerel}
\usepackage{braket}
\usepackage[normalem]{ulem}
\usepackage{float}

\usepackage{enumitem}

\usepackage[toc,page,header]{appendix}
\usepackage{minitoc}

\newcommand{\mL}{\mathcal{L}}

\usepackage[utf8]{inputenc}
\usepackage[english]{babel}
\usepackage{amsthm}
\newtheorem{definition}{Definition}[section]
\newtheorem{lemma}{Lemma}[section]

\newtheorem{proposition}[lemma]{Proposition}

\usepackage{amsthm}
\usepackage{amssymb}

\newcommand{\figeq}[2]{\vcenter{\hbox{  \includegraphics[scale = #1]{#2}}}}

\newcommand{\vx}{\xvec}

\usepackage{float}

\newcommand{\dk}{D_\mathcal{K}}

\graphicspath{{./Figures/}{./}}

\newcommand{\rmax}{r_\mathrm{max}}

\newcommand{\mk}{\mathcal{K}}

\newcommand{\mA}{\mathcal{A}}

\newcommand{\mB}{\mathcal{B}}

\newcommand{\vxt}{\tilde{\mathbf{x}}}

\newcommand{\vl}{\mathbf{l}}
\newcommand{\vvl}{\vec{\mathbf{l}}}

\newcommand{\hd}{d}
\newcommand{\hr}{r}
\newcommand{\hb}{b}

\newcommand{\bd}{\bar{d}}
\newcommand{\br}{\bar{r}}
\newcommand{\bb}{\bar{b}}

\newcommand{\Smax}{S_\mathrm{max}}
\newcommand{\Amax}{A_\mathrm{max}}

\input{epsf}

\newcommand{\splus}{ S^{+}}
\newcommand{\sminus}{ S^{-}}
\newcommand{\sz}{ S^{z}}

\newcommand{\vv}{\psi }
\newcommand{\bvv}{ \bar{\psi}}

\newcommand{\verdimers}{
\hspace{-0.12cm}
\vcenter{
\hbox{  \includegraphics[scale = 0.07]{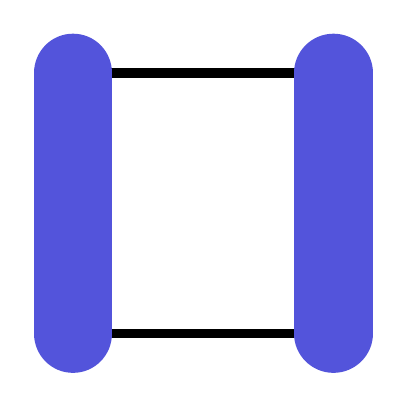}
}
}\hspace{-0.12cm}
}
\newcommand{\hordimers}{
\hspace{-0.12cm}
\vcenter{
\hbox{  \includegraphics[scale = 0.07]{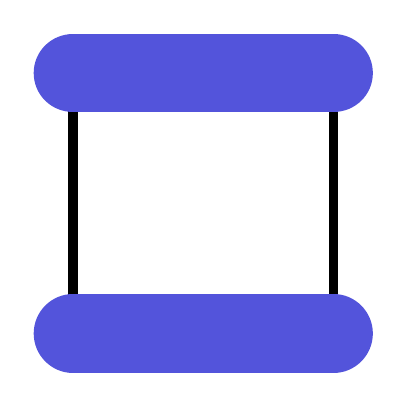}
}
}\hspace{-0.12cm}
}

\begin{document}

\doparttoc 
\faketableofcontents 

\title{Height-conserving quantum dimer models}

\author{Zheng Yan}
\affiliation{Department of Physics and HKU-UCAS Joint Institute of Theoretical and Computational Physics,
The University of Hong Kong, Pokfulam Road, Hong Kong, China}

\author{Zi Yang Meng}
\email{zymeng@hku.hk}
\affiliation{Department of Physics and HKU-UCAS Joint Institute of Theoretical and Computational Physics,
The University of Hong Kong, Pokfulam Road, Hong Kong, China}

\author{David A. Huse}
\affiliation{Department of Physics, Princeton University, Princeton, NJ 08544, USA}
\affiliation{Institute for Advanced Study, Princeton, NJ 08540, USA}
\author{Amos Chan}
\email{amos.chan@princeton.edu}
\affiliation{Princeton Center for Theoretical Science, Princeton University, Princeton, NJ 08544, USA}

\date{\today}

\begin{abstract}
We propose a height-conserving quantum dimer model (hQDM) such that the
lattice sum of its associated height field is conserved, and that it admits a Rokhsar-Kivelson (RK) point. 
The hQDM with minimal interaction range  on the square lattice exhibits Hilbert space fragmentation and 
maps exactly to the XXZ spin model on the square lattice in certain Krylov subspaces.
We obtain the ground-state phase diagram of hQDM via quantum Monte Carlo simulations, and demonstrate that a large portion of it is within the Krylov subspaces which admit the exact mapping to the
XXZ model, with dimer ordered phases
corresponding to easy-axis and easy-plane 
spin orders.
At the RK point, the apparent dynamical exponents obtained from the single mode approximation and the height correlation function %
show drastically different behavior across the Krylov subspaces, exemplifying Hilbert space fragmentation and emergent glassy phenomena.
\end{abstract}

\maketitle

\paragraph{Introduction.--} Quantum dimer models (QDM) \cite{rkham} are paradigmatic models of strongly-correlated systems subject to strong local constraints. 
Originally introduced to model the physics of short-range resonating valence bond states %
\cite{anderson1987rvb, anderson1974,Kivelson1987}, QDM have subsequently been shown to host  a plethora of phenomena, such as topological order and fractionalization~\cite{moessner2001tri,qdm2002kagome, moessner2001short}, mapping to height models~\cite{Nienhuis_1984, henley1997relaxation, henley2004}, unconventional phase transition with anyon condensation~\cite{Ralko2005,federic_qd_2006,Plat2015z2,ZhengYan2021Triangle,ZY2022}, emergent continuous symmetry and gauge field~\cite{ZhengYan2021Triangle,ZY2022} and more \cite{punk_qdm_2015, oakes_dimers_2018, powell_constraints_2018, feldmeier2019emergent, ryu2019,  alet2020, 
	flicker2021, wildeboer2021qdmscar, powell2022}.
On the other hand, recent experimental progress in ultracold atomic gases~\cite{Semeghini21,Satzinger2021,Samajdar2021,ZY2022} has led to a surge of interest in understanding non-equilibrium dynamics in
closed quantum many-body systems, 
including quantum many-body chaos~\cite{nahum2017, nahum2018, keyserlingk2018, cdc1,cdc2, Chan_2019, kitaev2015}, many-body localization~\cite{gornyi2005, baa2006, pal2010many, nandkishore2015many}, quantum many-body  scars~\cite{shiraishi2017systematic, moudgalya2018a, turner2018, moudgalya2018b, turner2018quantum, ho2018periodic, khemani2019int}, anomalous dynamics and subdiffusive behaviour~\cite{feldmeier2020, Scherg2021, morningstar2020, moudgalya2021spectral, iaconis2021multipole, Feldmeier2021slow,  Rakovszky2020sliom, aidelsburger2021tilt, waseem2020tilted}, localization in fractonic systems~\cite{pai2019, fractonreview,fractonreview2,PY2021fracton,zhoufracton2022} and Hilbert space fragmentation (HSF)~\cite{sala2020ergodicity, khemani2020, moudgalya2019krylov, yang2020HSF, lee2021HSF, moudgalya2021hilbert, mukherjee2021hsf, langlett2021hsf, Pozsgay2021HSF, herviou2021hsf, khudorozhkov2021hilbert, feng2022hsf, Li2021HSF, Patil2020}, the latter of which provide examples of non-integrable systems which fail to obey the Eigenstate Thermalization Hypothesis~\cite{deutsch1991quantum, srednicki1994chaos, rigol2008thermalization}. 

In this paper, guided by the connection between constrained dynamics and height-conserving field theories uncovered in \cite{moudgalya2021spectral}, we design a height-conserving quantum dimer model (hQDM) that displays phenomena of constrained systems, in particular HSF, yet retains attractive features of QDM described above.
Specifically, for the square-lattice hQDM with minimal-range interactions,
we show that the conservation of %
its associated height field leads to HSF, and that there exists an exact mapping between hQDM and a XXZ spin model in certain Krylov subspaces (KS). 
We employ the unbiased sweeping cluster quantum Monte Carlo (QMC) simulation~\cite{Yan_2019_mc, yan2021improved,ZhengYan2021Triangle,ZY2022} to obtain the ground-state phase diagram and find the main part of phase diagram can be characterised by that of the two-dimensional XXZ model, with different dimer ordering corresponding to easy-axis and easy-plane spin long-range orders.  At the RK point, the apparent dynamical exponents obtained from the single mode approximation and the height correlation of QMC data, exhibit drastically different behavior across the KS. These results confirm the fragmented Hilbert space and emergent glassy behavior of our model and its relevance to further developments of constrained quantum lattice models is discussed.

\paragraph{Model.--}
Consider close-packed tiling of dimers on a square (bipartite) lattice, such that each site is occupied by exactly one dimer. 
We introduce a height-conserving QDM such that (i) it admits a general Rokhsar-Kivelson (RK) point~\cite{RK1988} in its parameter space, and (ii) the sum
of its associated height field $h(\xvec)$ is conserved.
We focus on the Hamiltonian $H = H_{\mathrm{res}} + H_{\mathrm{diag}}$ with interaction of \textit{minimal} range on the square lattice, where
\begin{equation} \label{eq:hQDM_def}
\begin{aligned}
\Hres
=
- t \sum
\left(
\,
\left|
\vcenter{\hbox{  \includegraphics[scale = 0.1]{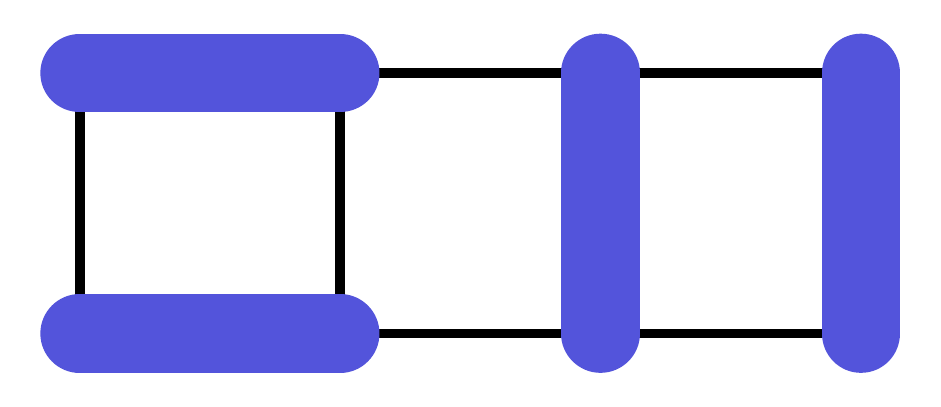}
}}
\right\rangle
\left\langle
\vcenter{\hbox{  \includegraphics[scale = 0.1]{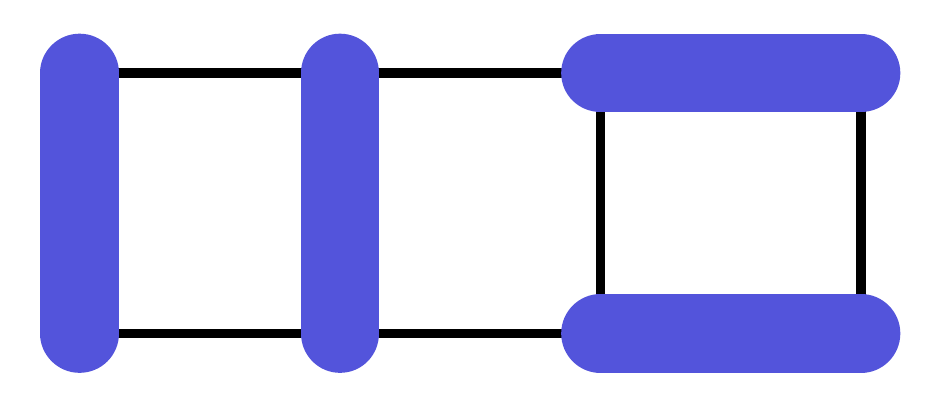}
}}
\right|
\right.
\\
\left.
+
\left|
\vcenter{\hbox{  \includegraphics[scale = 0.1]{2020_Quantum_Dimer_2_color_v1}
}}
\right\rangle
\left\langle
\vcenter{\hbox{  \includegraphics[scale = 0.1]{2020_Quantum_Dimer_1_color_v1}
}}
\right|
\,
\right) \;,
\\
\Hdiag
=
V \sum \left(
\,
\left|
\vcenter{\hbox{  \includegraphics[scale = 0.1]{2020_Quantum_Dimer_1_color_v1}
}}
\right\rangle
\left\langle
\vcenter{\hbox{  \includegraphics[scale = 0.1]{2020_Quantum_Dimer_1_color_v1}
}}
\right|
\right.
\\
\left.
+
\left|
\vcenter{\hbox{  \includegraphics[scale = 0.1]{2020_Quantum_Dimer_2_color_v1}
}}
\right\rangle
\left\langle
\vcenter{\hbox{  \includegraphics[scale = 0.1]{2020_Quantum_Dimer_2_color_v1}
}}
\right| \,
\right) \;.
\end{aligned}
\end{equation}
The sums are over all possible vertically- or horizontally-aligned next-to-nearest-neighbour plaquette pairs.
hQDM with interactions of longer range and/or conserved height multipoles (see Supplementary Materials (SM)~\cite{SM}) are amenable to field-theoretical analysis, which we discuss in detail in a forthcoming work~\cite{paper2}. 
  \begin{figure}[t]
   \includegraphics[width=0.46 \textwidth  ]{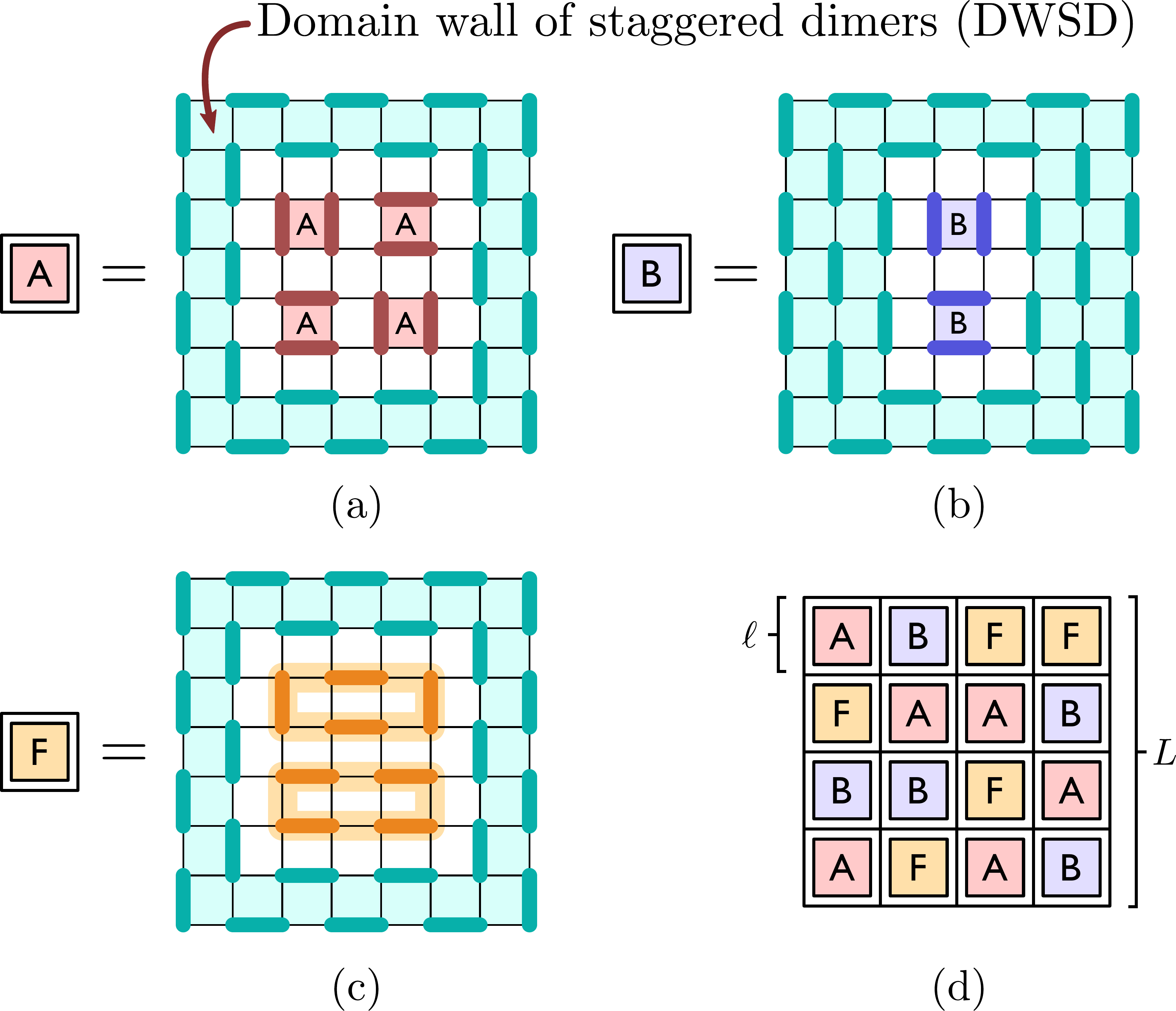}
    \caption{ 
    (a,b) DWSD (in cyan) surrounding regions of flippable plaquettes residing in sublattices A and B respectively.
    (c) DWSD surrounding a region of inactive dimer configuration. 
    Loop updates on the orange loops will generate additional inactive configurations exponential in the region size.
    (d) $O(\gamma^{L^2})$ number of KS with isolated active regions can be constructed by piecing (a,b,c) together with $\gamma >1$. 
    } \label{fig:subsector_puz}
\end{figure}
Like the QDM, the hQDM preserves the \textit{tilts} or \textit{winding numbers}, defined in terms of the height field as 
$t_1= [h(\xvec= (L,x_2))- h(\xvec= (0,x_2))]/L$ and 
$ t_2= [h(\xvec= (x_1,L))- h(\xvec= (x_1,0))]/L$ 
for any values of $x_1$ and $x_2$.
Additionally, the short-range nature of the pairwise plaquette flips of $\Hres^{(0)}$ leads to extra conserved quantities, namely the total height $I_X$ for each sublattice,
$\{X=\mathrm{A}, \mathrm{B}, \mathrm{C}, \mathrm{D} \}$, shown in Fig.~\ref{fig:phasediag}(b).
Together, the set of conserved charges is $\{ t_1, t_2, I_{ \mathrm{A}}, I_{ \mathrm{B}}, I_{ \mathrm{C}}, I_{ \mathrm{D}} \}$. 

\begin{figure}[t]
\centering
   \includegraphics[width=\columnwidth  ]{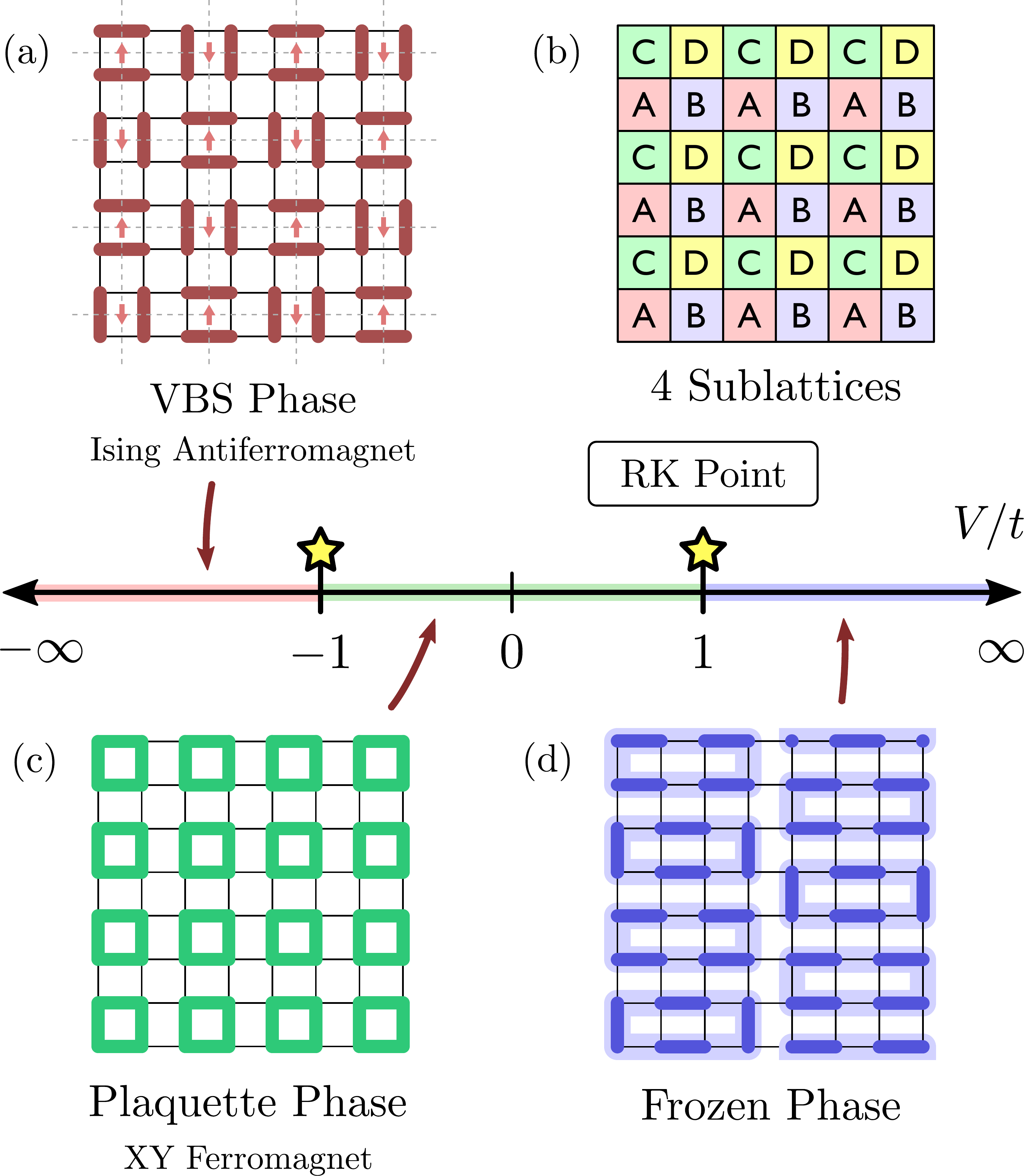}
    \caption{Phase diagram of the minimal-range hQDM 
    on the square lattice. 
    For $V/t \in (-\infty, -1)$, the system is in a VBS phase, with eight degenerate ground states generated by performing  plaquette-flips and translations of the pattern (a) to four sublattices defined in (b). The hQDM restricted to the KS of (a) can be exactly mapped to the XXZ model, where the VBS is the Ising antiferromagnetic phase. For $V/t \in (-1, 1)$, the system is in the plaquette phase illustrated in (c), corresponding to the XY ferromagnetic phase of the XXZ model. 
    For $V/t \in (1, \infty)$, the system is in the frozen phase and has  $O(e^{L^2})$ number of inactive and degenerate GS, some of which can be generated by performing loop updates on (d). 
    The transition point between VBS and plaquette phases is equivalent to the antiferromagnetic Heisenberg point of the XXZ model and that between the plaquette and frozen phases is the RK point.
    } \label{fig:phasediag}
\end{figure}

\paragraph{Hilbert space fragmentation.--} Our hQDM has HSF, i.e., the full space of states breaks into dynamically disconnected KS.  A \textit{Krylov subspace} (KS) is defined as $\mk(H, \ket{\psi})  = \mathrm{span}\{ H^n\ket{\psi}, n\in \mathbb{N} \} \equiv \mk( \ket{\psi})$ with size $D_\mk$, and in our context, the basis states $\ket{\psi}$ are dimer configurations. 
A system is said to exhibit HSF~\cite{sala2020ergodicity, khemani2020, moudgalya2019krylov} if $\mk(H, \ket{\psi})$ does not span the Hilbert subspace labelled by the symmetry quantum numbers of $\ket{\psi}$.
In our minimal-range hQDM, 
we prove in SM~\cite{SM} that there exists an exponential number (in system size) of subsectors with (i) $\dk=1$ (frozen KS), (ii) $\dk=O(1)$ (minimally active KS), and (iii) $\dk=O(\alpha^{L^2})$ (active KS) with constant $\alpha >1$.
The idea of the proof is demonstrated in a construction using \textit{domain wall of staggered dimers} (DWSD) (see Fig.~\ref{fig:subsector_puz}), defined as a connected region of plaquettes where no plaquettes are occupied by a parallel dimer pair. 
A DWSD is a ``blockade'' since regions separated by DWSD are dynamically disconnected in that no terms in $H$ can act non-trivially on or across the DWSD.
Claims (i-iii) are then proven by obtaining distinct KS via piecing multiple blocks of active or inactive dimer configurations surrounded by DWSD, shown in Fig.~\ref{fig:subsector_puz} and SM~\cite{SM}.  We also have strong indications, but not a proof, as discussed below, that this minimal-range hQDM has strong HSF, so no KS has an entropy density equal to the full entropy density.

\paragraph{Phase diagram.--}
We can obtain the phase diagram (Fig.~\ref{fig:phasediag}) of our minimal-range hQDM 
in several tractable limits.%
For $V/t \to -\infty$, the Hamiltonian favors states with the most number of flippable plaquette pairs. Such a state is shown in Fig.~\ref{fig:phasediag}(a), and 
we prove in SM~\cite{SM} that there are eight degenerate ground states (GS) generated by (i) performing pair-dimer-rotation on all plaquettes with parallel dimers, and (ii) translating the pattern [Fig.~\ref{fig:phasediag}(a)] to the four sublattices [Fig.~\ref{fig:phasediag}(b)]. 
Furthermore, these eight degenerate GS belong to four (disconnected) KS associated with each sublattice~\cite{SM}.
We refer to these GS as the valence bond solid (VBS).

For $V/t \in (1, \infty)$, the Hamiltonian favors states without any flippable plaquette pairs, i.e. the system is in the frozen phase. 
We explicitly construct degenerate GS with extensive entropy in Fig.~\ref{fig:phasediag}(d), generated by performing loop updates -- by flipping occupied links to unoccupied links in a given loop, and vice versa -- on each blue loop~\cite{SM}.
Each of such states form a KS of size one. 
Note that all inactive configurations in QDM are also valid  inactive configurations in hQDM.
However, unlike QDM, hQDM has extensive GS entropy in the frozen phase while QDM has sub-extensive GS entropy in the staggered phase for $V/t \in(1,\infty)$, 
and the degenerate hQDM GS contain states with a variety of height tilts, including the columnar states with zero tilt.

\begin{figure}[t!]
\centering
\includegraphics[width=\columnwidth]{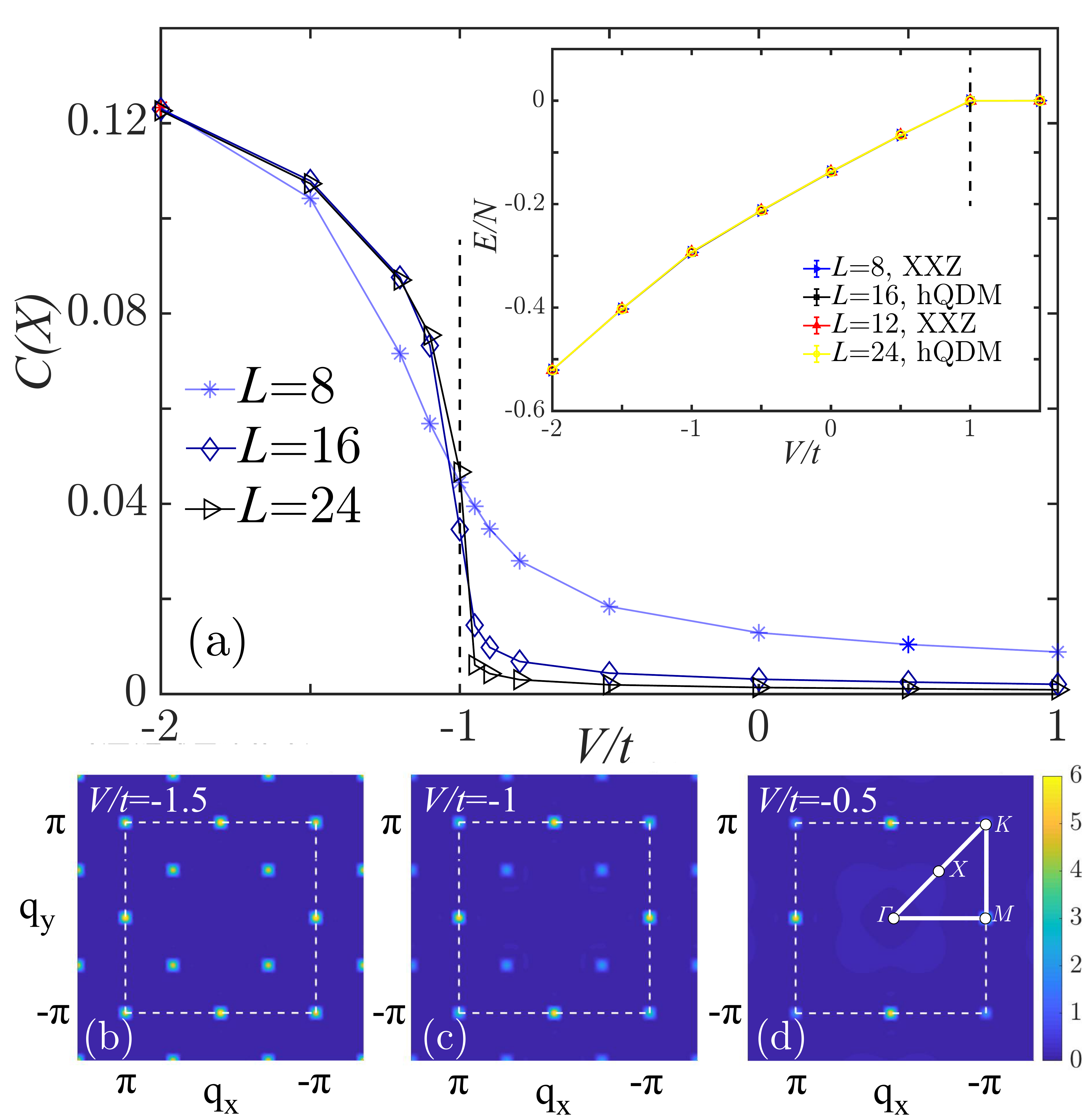}
    \caption{
(a) The dimer-pair structure factor $C(X)$ for hQDM shows a first order transition from VBS (Ising antiferromagnet in the XXZ model) to the Plaquette (XY ferromagnet in the XXZ model) phase at $V=-1$. Inset shows energy density of hQDM with $L=16,24$ and of XXZ model with $L=8,12$ as a function of $V/t$. The curves coincide and exhibit transition at the RK point $V=1$. (b), (c) and (d) are the dimer-pair structure factor at $V=-1.5$, -1 and -0.5 inside the Brilliouin zone, respectively. With $C(\mathbf{q})$ peaks at $X$ and $M$ in (b) and (c), but only peak at $M$ in (d). The high-symmetry path is denoted in (d).
} \label{Fig2}
\end{figure}

At $V/t=1$, as in QDM, hQDM has a RK point \cite{RK1988, henley1997relaxation, henley2004}. 
The GS at RK point are highly degenerate, with a unique GS from each KS given by
 $  |\GS\rangle_{\mk} \propto \sum_{C \in \mk } |C\rangle$.
The field theoretical approaches at the RK point will be discussed in detail in~\cite{paper2}.
Away from the RK point, one can consider the following heuristic argument to justify the frozen and VBS phases: 
with $V/t  \rightarrow 1^+$, 
subspaces without any flippable plaquettes have 
the lowest energy, consistent with the frozen phase;
with $V/t  \rightarrow 1^-$, 
the lowest energy subspaces should be the ones with the most flippable plaquettes, i.e. the KS where VBS reside.

There exists an exact mapping between our minimal-range hQDM  
and XXZ spin models in a certain set of KS, $\mk(\ket{\psi}_X)$, where all flippable plaquettes in $\ket{\psi}_X$ are residing on the same sublattice $X \in \{\mathrm{A}, \mathrm{B}, \mathrm{C}, \mathrm{D} \}$ via 
\be \label{eq:map}
\; 
    \left|
\hspace{-0.05cm}
\vcenter{\hbox{  \includegraphics[scale = 0.1]{2020_05_Quantum_Dimer_Color_plaquette_2}
}}
\hspace{-0.05cm}
\right\rangle_\vx
\xrightarrow{\; \; g \; \; } \, 
    \left|
\downarrow
\right\rangle_{\vxt}
\, ,
\qquad 
    \left|
\hspace{-0.05cm}
\vcenter{\hbox{  \includegraphics[scale = 0.1]{2020_05_Quantum_Dimer_Color_plaquette_1}
}}
\hspace{-0.05cm}
\right\rangle_\vx
\xrightarrow{\; \; g \; \; }  \,
    \left|
\uparrow
\right\rangle_{\vxt}
\; ,
\ee
where $\vx$ and $\vxt$ labels the plaquettes of the original lattice and of sublattice $X$ respectively~\cite{SM}.
The validity of this mapping relies on the fact 
that flippable plaquettes in these KS always remain in the same sublattice under the action of $H$. Therefore, it maps a $2L\times 2L$ hQDM into a $L\times L$ XXZ model in certain subspaces:
\be\label{eq:xxz_mapping}
H_{\mathrm{hQDM}} 
\Big|_{
\mk
\left(
{\ket{\psi}_{X}}
\right)
, \{\vx\} 
}
=
H_{\mathrm{XXZ}}
\Big|_{
\mk
\left(
g\left( \ket{\psi}_X \right)
\right)
, \{\vxt\} 
} \;,
\ee
and
\be\label{eq:xxz_model}
\begin{split}
\left. H_{\mathrm{XXZ}}\right|_{ \{\vxt\} } = -t \sum_{\left\langle \vxt, \vxt' \right\rangle} \left( 
\splus_{\vxt} \sminus_{\vxt'}  + \sminus_{\vxt} \splus_{\vxt'}  
\right) 
\\
+
\frac{V}{2} \sum_{\left\langle \vxt, \vxt' \right\rangle} \left(1 -  \sz_{\vxt} \sz_{\vxt'}  \right) \;,
\end{split}
\ee
where $S^{a}_{\vxt}$ with $\alpha=x,y,z$ are the Pauli matrices, and $S^{\pm} = \frac{1}{2}(S^x \pm i S^y$). 
Note that when longer-range terms of hQDM are introduced, this mapping is no longer valid as the relevant KS have enlarged. 
However, for the minimal-range hQDM, this XXZ description is crucial in describing the phase diagram for $V/t\in (-\infty, 1)$.
\paragraph{Sweeping cluster quantum Monte Carlo on hQDM.--}
We employ the sweeping cluster QMC algorithm~\cite{Yan_2019_mc,yan2021improved} to simulate the phase diagram.
The method has recently been intensively applied on the square and triangle lattice QDM models to map out the GS phase diagram and extract the low-energy excitations~\cite{ZhengYan2021Mixed,ZhengYan2021Triangle,ZY2022}. 
The method not only respects the local constraint in hQDM between each QMC update, but also allow for loop updates with randomly-sampled loops, ensuring different subspaces labelled by $\{ t_1, t_2, I_{\mathrm{A}}, I_{\mathrm{B}}, I_{\mathrm{C}}, I_{ \mathrm{D}} \}$ and the fragmented KS within being sampled. We simulated the hQDM with linear system sizes up to $L=24$, and the inverse temperature $\beta=L$.
Our QMC data show that the GS in $V/t\in (-\infty, 1)$ indeed reside
in KS described by the mapping in Eq.~\eqref{eq:xxz_mapping}. 
To probe the phase diagram, we study the dimer-pair structure factor $C_{ij}(\mathbf{q})$, defined as the Fourier transform of the pair dimer correlation function
\begin{equation}\label{eq:cor_func}
    C_{ij,\xvec -\xvec '}=\frac{\langle \dop_{i\xvec}\dop_{j,\xvec'}\rangle-\langle \dop_{i,\xvec}\rangle\langle \dop_{j,\xvec'}\rangle}{
    \langle \dop^2_{i,\xvec'}\rangle-\langle \dop_{i,\xvec'}\rangle^2} \;,
\end{equation}
where $\dop_{i,\xvec}$ with $i=\hordimers,\verdimers$ is the projector of vertically- or horizontally-aligned dimer pairs  at plaquette $\xvec$~\cite{SM}.
The $C(\mathbf{q}=X)$ where $X=(\pi/2,\pi/2)$ and $i=j=\hordimers$ serves as the square of the order parameter for the VBS phase. As shown in Fig.~\ref{Fig2} (a), as $V/t$ increases from $-\infty$, the hQDM undergoes a first order phase transition at $V/t=-1$ from the VBS phase ($C(X)$ is finite) to a plaquette phase ($C(X)$ is zero) in one of the four sublattices. The same information is revealed in the $C(\mathbf{q})$ in the Brillouin zone (BZ) in Fig.~\ref{Fig2}(b), (c) and (d), where the structure peak at $X$ appears at $V=-1.5$ (we note here the peaks at $M$ and $K$ are due to the repetition of that at $X$) and $V=-1$ (here the peaks at $M$ and $X$ are due to the coexistence of VBS and plaquette phases), but at $V=-0.5$, when the system is inside the plaquette phase, $C(\mathbf{q})$ only develops peaks at $M=(\pi,0)$ points.

This transition in hQDM corresponds to the transition from antiferromagnetic Ising phase to XY phase in XXZ model in each of the KS corresponding to the four sublattices. At the Heisenberg point $V=-1$, the O(3) symmetry is recovered and the GS spontaneously breaking this symmetry resulting in the coexistence of peak at $X$ and $M$ in Fig.~\ref{Fig2}(c). 
Similar phenomena have been seen in the context of deconfined quantum critical point and dimerized quantum magnets~\cite{BWZhao2018,GYSun2021}.

Moreover, we also compare the GS energy density from sweeping cluster QMC on hQDM, with that from directed loop QMC simulation~\cite{syljuaasen2002quantum,syljuaasen2003directed,alet2005generalized} on the XXZ model. As shown in the inset of Fig.~\ref{Fig2}(a), two energies are identical, meaning the hQDM GS indeed resides in the XXZ KS, where the mapping to XXZ model is valid.
Interestingly, at $V/t\ge 1$, the system is at the RK point and frozen phase, our QMC data show the mapping remains valid. 
Indeed, the ferromagnetic phase of XXZ model (when $V/t > 1$) can be mapped to the columnar phase of the square lattice QDM at $V/t>1$~\cite{moessner2008review}, which is one of the degenerate GS in the frozen phase of hQDM. 

\begin{figure}[t!]
\centering
\includegraphics[width=\columnwidth]{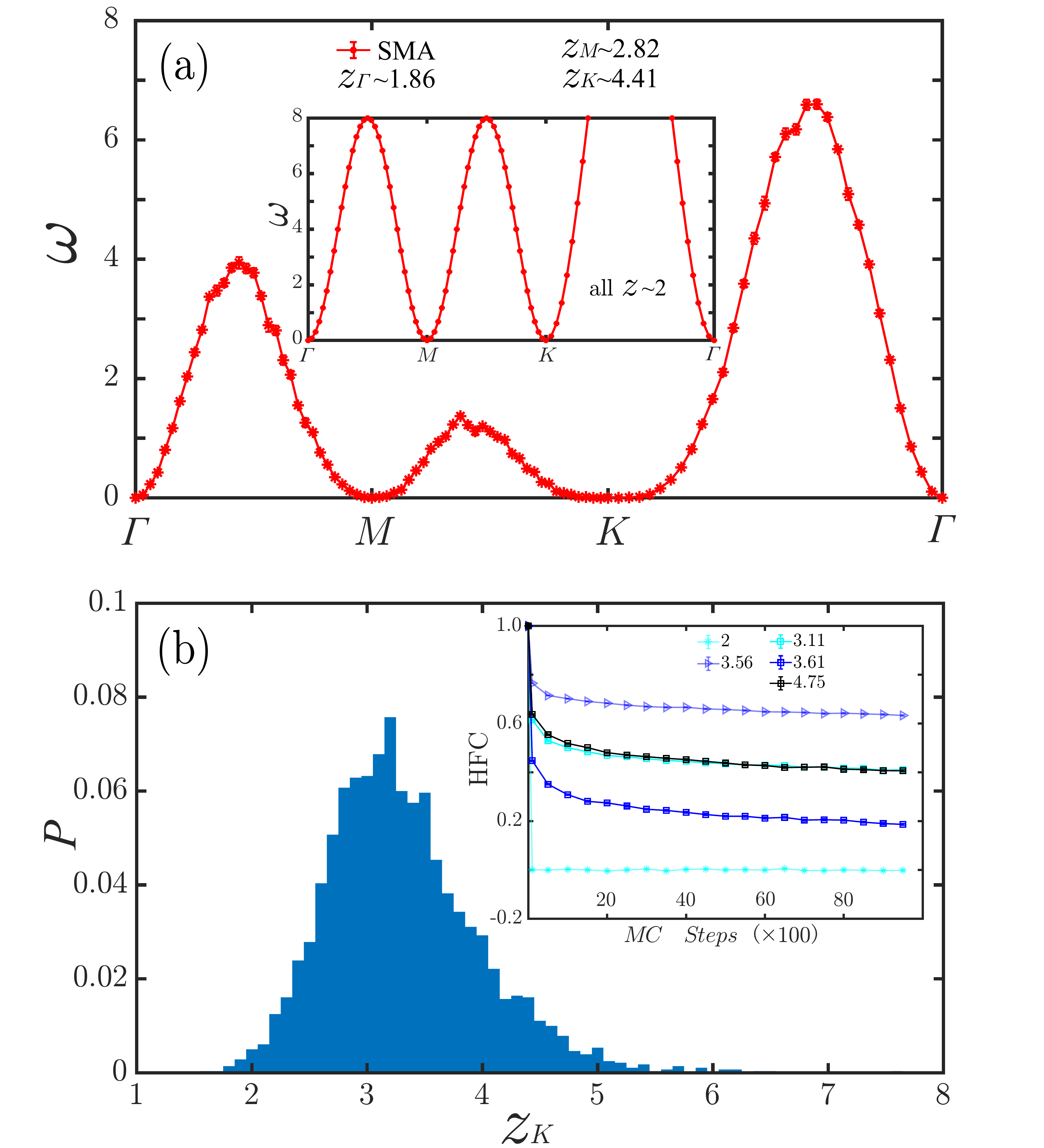}
\caption{(a). The dispersion of hQDM at the RK point in a random sector is given along the high-symmetry path, with the $z_\mathbf{q}$ denoted. Inset shows the dispersion in the XXZ sector, with all $z=2$. (b). The distribution of $z_K$ for hQDM via random walk MC and SMA, the distribution is obtained via 3000 randomly chosen sectors. Inset shows the real-space height field autocorrelation function of A sublattice in different KS with different $z_K$ (as labeled in the figure) against MC steps. 
}
 \label{Fig4}
\end{figure}

\paragraph{Dispersion and dynamic properties.--}
We finally study low-lying excitations at the RK point of the minimal-range hQDM where each sector has a zero-energy ground state, and show that the apparent dynamical exponent $z$ varies strongly between KS and the system develops emergent glassy behavior, indicating strong HSF.
%
%
%
%
%
%
%

We utilize the single mode approximation (SMA) upon QMC data to approximate a dimer dispersion~\cite{L_uchli_2008}. %
Besides the low-energy dispersion, via the random walk method, we also sample KS randomly and obtain the distribution of apparent dynamical exponents $z_\mathbf{q}$ at various $\mathbf{q}$ points in the BZ. 

Our results are shown in Fig.~\ref{Fig4}. Different from the QDM, the dynamics of hQDM are extremely sensitive to KS. In the XXZ KS which hosts the ground state for $V<t$, the dynamic exponent $z=2$ (similar to that of the QDM~\cite{L_uchli_2008}), all the exponents near $K$ are indeed close to 2 as shown in the inset of Fig.~\ref{Fig4} (a)~\footnote{Excitations in the XXZ KS can be seen as a Goldstone mode in ferromagnetic spin-$1/2$ Heisenberg model, the power $z=2$ can be obtained via a simple spin-wave approximation.}. However, the hQDM dispersion in other KS (the main panel of Fig.~\ref{Fig4}(a)) clearly acquire larger apparent $z$. 
We further plot the distribution of $z_{K=(\pi,\pi)}$ extracted from the dispersion from sampling different sectors in Fig.~\ref{Fig4}(b), which exhibits a wide distribution centered around $z=3$. 
Such behavior suggests strong fragmentation, since randomly-chosen states do not appear to be in the same KS with any significant probability.  We expect that the model will only have weak HSF upon inclusion of longer range terms in \eqref{eq:hQDM_def}.

We also study the autocorrelation function of height field at RK points using classical MC, by updating dimer configurations via the $\Hres$ on a randomly chosen position from a initial state of certain sector. 
The measurement of height autocorrelation along Monte Carlo steps (MCS) in one sublattice is shown in the inset of Fig.~\ref{Fig4}(b). 
In the XXZ sector, $z_K \sim 2$, autocorrelation decays very fast with ergodicity in the related Hilbert subspace. 
In other sectors, $z_K > 2$, the autocorrelation decays much slower which hinders the relaxation. There seems to be no obvious relationship between the autocorrelation decay and $z_K$. It is a strong evidence for the glassy behavior which emerges in most sectors of hQDM.

\paragraph{Conclusion and outlook.--} In this paper, we introduced the hQDM as a realization of height-conserving models. Rich phenomena such as HSF with DWSD as blockades, mapping to XXZ model, and the emergent glassy behavior at the RK are observed analytically and numerically.
The hQDM height field theory~\cite{henley1997relaxation,henley2004}, lattice gauge theory, Lifshitz models and conformal critical point~\cite{ardonne2004} will be presented in upcoming work~\cite{paper2}.
Emergent fractonic behaviour~\cite{haah, vijay2015, vijay2016, chamon2005, pretko2017a, pretko2017b}, explored recently in dimer models in higher than two dimensions~\cite{Feldmeier2021fracdimer, you2021fracdimer}, 
the nature of the low-lying excitations, and the structure of HSF in hQDM with interaction of various ranges are also interesting open directions. 

\paragraph{Acknowledgements.--}
We wish to thank Jie Wang and Shivaji Sondhi for useful discussions, and Sanjay Moudgalya and Abhinav Prem for previous collaboration. 
ZY and ZYM acknowledge support from the RGC of Hong Kong SAR of China (Grant Nos. 17303019, 17301420, 17301721 and AoE/P-701/20), the K. C. Wong Education Foundation (Grant No. GJTD-2020-01) and the Seed Funding "Quantum-Inspired explainable-AI" at the HKU-TCL Joint Research Centre for Artificial Intelligence. We thank the Information Technology Services at the University of Hong Kong and the Tianhe platforms at the National Supercomputer Centers in Guangzhou for technical support and generous allocation of CPU time. The authors acknowledge Beijng PARATERA Tech CO.,Ltd.(\url{https://www.paratera.com/}) for providing HPC resources that have contributed to the research results reported within this paper.
DAH is supported in part by NSF QLCI grant OMA-2120757.  AC is supported by fellowships from the Croucher foundation and the PCTS at Princeton University.

\bibliography{cQDM}


\onecolumngrid
\newpage 

\appendix

\setcounter{equation}{0}
\setcounter{figure}{0}
\renewcommand{\thetable}{S\arabic{table}}
\renewcommand{\theequation}{S\thesection.\arabic{equation}}
\renewcommand{\thefigure}{S\arabic{figure}}
\setcounter{secnumdepth}{2}

\begin{center}
{\Large Supplementary Material \\ 
\vspace{0.2cm}
Height-conserving Quantum Dimer Models
}
\end{center}

In this supplementary material we provide additional details about:
\begin{enumerate}[label=\Alph*]
    \item Height representation of dimer configurations
    \item Construction of height-conserving qantum dimer models (hQDM)
    \begin{itemize}
        \item[1. ] hQDM in terms of height-flip operators   
        \item[2. ] A particular short-range hQDM  with general $m$
        \item[3. ] An example of longer-range hQDM with $m=0$
    \end{itemize}
    \item Analytically tractable limits in the phase diagram
        \begin{itemize}
        \item[1. ]  Ground state as $V/t \rightarrow \infty$ 
        \item[2. ]  Ground state as $V/t \rightarrow  - \infty$ 
    \end{itemize}
    \item Structure of Hilbert space
    \begin{itemize}
        \item[1. ] Domain walls of staggered dimers (DWSD) as blockades
        \item[2. ]  Exponential number of Krylov subspace of size $1$ 
        \item[3. ]  Exponential number of Krylov subspace of size $O(1)$ 
        \item[4. ] Exponential number  of Krylov subspace of size $O(\alpha^{L^2})$ with $\alpha>1$
    \end{itemize}
    \item Mapping between hQDM and spin models
    \item Dimer pair correlation functions and structure factors
    \item Single mode approximation
\end{enumerate}



\section{Height representation of dimer configurations}
\begin{figure}[H]
\centering
\includegraphics[width=0.8 \textwidth  ]{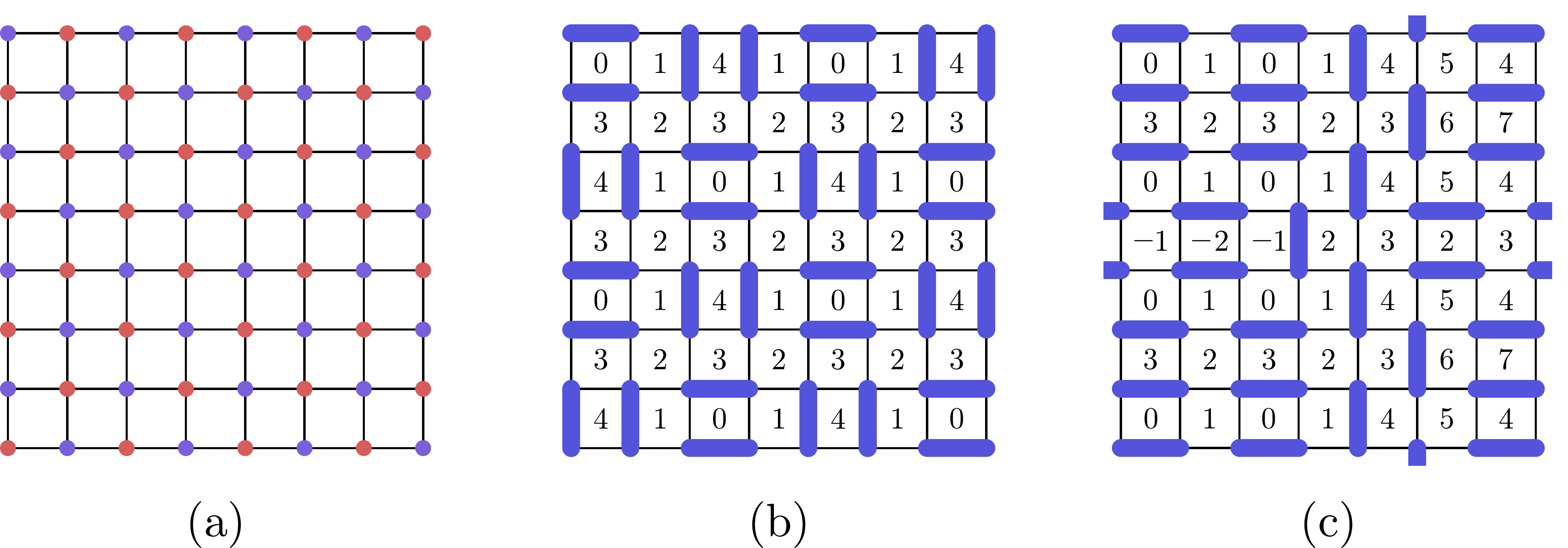}
    \caption{(a) The square lattice is a bipartite lattice whoses sites can be partitioned into set A (red) and B (blue), such that sites in set A are the only nearest neighbours of sites in set B.
    (b-c): Height representations of two full-packed dimer configurations: 
    (b) a ground state in the VBS phase, and
    (c) a ground state in the frozen phase (with periodic boundary condition) 
    }
\end{figure}

Consider a closed-pack dimer configuration on a bipartite lattice, taken to be a square lattice for simplicity.
Partition the set of all sites into two sets, set A and set B, such that sites in set A are the only nearest neighbours of sites in set B (and vice versa).
For each close-packed dimer configuration, we can associate a height representation as follows:
Assign an arbitrary integer value of height to  an arbitrary plaquette, say height $h(\xvec)=0$ to the top left plaquette. 
Pick one of the sites at the corners of this plaquette, say the bottom right corner. 
If the site is from the B (A) sublattice, then we walk to another plaquette in the (counter-)clockwise direction around the bottom right site. 
Suppose the old plaquette has height $ h(\xvec')$.
We assign a height $h(\xvec'')= h(\xvec')+1$ to the new plaquette if there is not a dimer in between the two plaquettes, and $h(\xvec'')= h(\xvec')-3$.   
We apply this rule to all plaquettes until all plaquettes are assigned with some height values. 
Note that the height field assignment is non-unique, since we can pick an arbirary integer value to the initial plaquette.
%
%



\section{Construction of height-conserving qantum dimer models (hQDM)}\label{app:hQDM_overall}
We construct the hQDM model based on two criteria: 
\begin{itemize}
    \item[(i)] The hQDM model conserves the $m$-th multipole moments of its associated height field $h(\xvec)$, i.e. $[H^{(m)},I^{a_1 \dots a_m}_{m, \mL}]= 0 $, where $I^{a_1 \dots a_m}_{m, \mL}$ is defined as
\be \label{eq:height_charge}
I_{m,R}^{a_1 \dots a_m} = \sum_{\xvec \in  R} x_{a_1} \dots x_{a_m} h(\xvec) \;,
\ee
where $a_i = 1,2$ denotes the two spatial dimensions, and the sum of $\xvec$ is over plaquettes inside some region $R$.
    \item[(ii)] The model  admits a Rokhsar-Kivelson point, i.e. the Hamiltonian is frustration-free at some points in the parameter space.
\end{itemize}
In Appendix \ref{app:hQDM_heightflip}, we state a definition of hQDM of general interaction range in terms of height-flipping operators. 
In Appendix \ref{app:higher_m}, we define a particular short-range hQDM for general $m$, explicitly state it for $m=0, 1$ and $2$, and show that this short-range hQDM is in fact minimal-range hQDM for $m=0$.
In Appendix \ref{app:hQDM_longer}, we state an example of longer-range hQDM for $m=0$, which contains ``knight moves''.

\subsection{hQDM in terms of height-flip operators   }\label{app:hQDM_heightflip}

Consider the bipartite square lattice with sublattices labelled by $\mA$ and $\mB$. We define the plaquette state and the $\overline{(\dots)}$ operation on the plaquette state as 
\be 
\vv_\vx \equiv 
\left\{
\begin{aligned}
\;
&
    \left|
\hspace{-0.05cm}
\vcenter{\hbox{  \includegraphics[scale = 0.1]{2020_05_Quantum_Dimer_Color_plaquette_1}
}}
\hspace{-0.05cm}
\right\rangle_\vx
\quad \quad
&
\text{if }\vx \in \mA 
\\
\; 
&
    \left|
\hspace{-0.05cm}
\vcenter{\hbox{  \includegraphics[scale = 0.1]{2020_05_Quantum_Dimer_Color_plaquette_2}
}}
\hspace{-0.05cm}
\right\rangle_\vx
&
\text{if }\vx \in \mB 
\end{aligned}
\right.
\;,
\qquad \qquad 
\bvv_\vx \equiv
\left\{
\begin{aligned}
\; 
&
    \left|
\hspace{-0.05cm}
\vcenter{\hbox{  \includegraphics[scale = 0.1]{2020_05_Quantum_Dimer_Color_plaquette_2}
}}
\hspace{-0.05cm}
\right\rangle_\vx
\quad \quad
&
\text{if }\vx \in \mA 
\\
\; 
&
    \left|
\hspace{-0.05cm}
\vcenter{\hbox{  \includegraphics[scale = 0.1]{2020_05_Quantum_Dimer_Color_plaquette_1}
}}
\hspace{-0.05cm}
\right\rangle_\vx
\quad \quad 
&
\text{if }\vx \in \mB 
\end{aligned}
\right.
\;\; .
\ee
Furthermore, we define the operators 
\be \label{eq:heigh_flip1}
\hr_\vx \equiv
\vv_\vx  \bvv_\vx^\dagger 
=
\left\{
\begin{aligned}
\; 
&
    \left|
\hspace{-0.05cm}
\vcenter{\hbox{  \includegraphics[scale = 0.1]{2020_05_Quantum_Dimer_Color_plaquette_1}
}}
\hspace{-0.05cm}
\right\rangle_\vx
\left\langle
\hspace{-0.05cm}
\vcenter{\hbox{  \includegraphics[scale = 0.1]{2020_05_Quantum_Dimer_Color_plaquette_2}
}}
\hspace{-0.05cm}
\right|_\vx
\quad 
&
\text{if }\vx \in \mA 
\\
\; 
&
    \left|
\hspace{-0.05cm}
\vcenter{\hbox{  \includegraphics[scale = 0.1]{2020_05_Quantum_Dimer_Color_plaquette_2}
}}
\hspace{-0.05cm}
\right\rangle_\vx
\left\langle
\hspace{-0.05cm}
\vcenter{\hbox{  \includegraphics[scale = 0.1]{2020_05_Quantum_Dimer_Color_plaquette_1}
}}
\hspace{-0.05cm}
\right|_\vx
\quad 
&
\text{if }\vx \in \mB 
\end{aligned}
\right.
\;,
\qquad 
\hd_\vx \equiv
\vv_\vx  \vv_\vx^\dagger 
=
\left\{
\begin{aligned}
\; 
&
    \left|
\hspace{-0.05cm}
\vcenter{\hbox{  \includegraphics[scale = 0.1]{2020_05_Quantum_Dimer_Color_plaquette_1}
}}
\hspace{-0.05cm}
\right\rangle_\vx
\left\langle
\hspace{-0.05cm}
\vcenter{\hbox{  \includegraphics[scale = 0.1]{2020_05_Quantum_Dimer_Color_plaquette_1}
}}
\hspace{-0.05cm}
\right|_\vx
\quad 
&
\text{if }\vx \in \mA 
\\
\; 
&
    \left|
\hspace{-0.05cm}
\vcenter{\hbox{  \includegraphics[scale = 0.1]{2020_05_Quantum_Dimer_Color_plaquette_2}
}}
\hspace{-0.05cm}
\right\rangle_\vx
\left\langle
\hspace{-0.05cm}
\vcenter{\hbox{  \includegraphics[scale = 0.1]{2020_05_Quantum_Dimer_Color_plaquette_2}
}}
\hspace{-0.05cm}
\right|_\vx
\quad 
&
\text{if }\vx \in \mB 
\end{aligned}
\right.
\;\; ,
\ee
and their $\overline{(\dots)}$ operation is given by
\be \label{eq:heigh_flip2}
\br_\vx \equiv
\bvv_\vx  \vv_\vx^\dagger 
=
\left\{
\begin{aligned}
\; 
&
    \left|
\hspace{-0.05cm}
\vcenter{\hbox{  \includegraphics[scale = 0.1]{2020_05_Quantum_Dimer_Color_plaquette_2}
}}
\hspace{-0.05cm}
\right\rangle_\vx
\left\langle
\hspace{-0.05cm}
\vcenter{\hbox{  \includegraphics[scale = 0.1]{2020_05_Quantum_Dimer_Color_plaquette_1}
}}
\hspace{-0.05cm}
\right|_\vx
\quad 
&
\text{if }\vx \in \mA 
\\
\; 
&
    \left|
\hspace{-0.05cm}
\vcenter{\hbox{  \includegraphics[scale = 0.1]{2020_05_Quantum_Dimer_Color_plaquette_1}
}}
\hspace{-0.05cm}
\right\rangle_\vx
\left\langle
\hspace{-0.05cm}
\vcenter{\hbox{  \includegraphics[scale = 0.1]{2020_05_Quantum_Dimer_Color_plaquette_2}
}}
\hspace{-0.05cm}
\right|_\vx
\quad 
&
\text{if }\vx \in \mB 
\end{aligned}
\right.
\;,
\qquad 
\bd_\vx \equiv
\bvv_\vx  \bvv_\vx^\dagger 
=
\left\{
\begin{aligned}
\; 
&
    \left|
\hspace{-0.05cm}
\vcenter{\hbox{  \includegraphics[scale = 0.1]{2020_05_Quantum_Dimer_Color_plaquette_2}
}}
\hspace{-0.05cm}
\right\rangle_\vx
\left\langle
\hspace{-0.05cm}
\vcenter{\hbox{  \includegraphics[scale = 0.1]{2020_05_Quantum_Dimer_Color_plaquette_2}
}}
\hspace{-0.05cm}
\right|_\vx
\quad 
&
\text{if }\vx \in \mA 
\\
\; 
&
    \left|
\hspace{-0.05cm}
\vcenter{\hbox{  \includegraphics[scale = 0.1]{2020_05_Quantum_Dimer_Color_plaquette_1}
}}
\hspace{-0.05cm}
\right\rangle_\vx
\left\langle
\hspace{-0.05cm}
\vcenter{\hbox{  \includegraphics[scale = 0.1]{2020_05_Quantum_Dimer_Color_plaquette_1}
}}
\hspace{-0.05cm}
\right|_\vx
\quad 
&
\text{if }\vx \in \mB 
\end{aligned}
\right.
\;\; .
\ee
Note that the $\overline{(\dots)}$ operation is distributed without altering the order of the objects within the brackets.
We define the operators 
\be \label{eq:b_op}
\begin{aligned}
\hb^{(-1)}(\vx) =& \;\hb_\vx  
\\
\hb^{(0)}(\vx; (\vl_0)) =& \;
\hb^{(-1)}(\vx)
\; \bb^{(-1)}(\vx+ \vl_0)
\\
\hb^{(1)}(\vx; (\vl_0,\vl_1)) =&\;
\hb^{(0)}(\vx;(\vl_0))\;
\bb^{(0)}(\vx+\vl_1; (\vl_0+\vl_1))
\\
\vdots
\\
\hb^{(m)}(\vx; \vvl_m)=&\;
\hb^{(m-1)}
(\vx; \vvl_{m-1}) \;
\bb^{(m-1)}
(\vx +\vl_m ; \vvl_{m-1} \oplus \vl_m )
\;,
\end{aligned}
\ee
where  $\vvl_m =(\vl_0,\vl_1, \dots, \vl_m)$ parametrizes the plaquette distances between height flip operators, and we define the operation $\vec{\mathbf{a}}
\oplus \mathbf{b} = (\mathbf{a}_1+ \mathbf{b}, \mathbf{a}_2+ \mathbf{b}, \dots)$.
The symbol $\hb$ can take the values of $\hr$, $\hd$, and $\vv$. For $\hb=\hr$, the operator $\hb^{(m)}(\vx; \vvl_m)$ can be interpretted as a $m$-th multipole creation or annihilation operator. 
The Hamiltonian can now be written compactly as 
\be\label{eq:hQDM_gen_def}
\begin{split}
H^{(m)} = & \, \Hres^{(m)} + \Hdiag^{(m)}
\\
\Hres^{(m)}
=&
-t 
 \sum_{\vx; \, \vvl:\, r[S(\vx; \vvl)] \leq \rmax } \left[
\,
\hr^{(m)}(\vx; \vvl_m) +
\br^{(m)}(\vx; \vvl_m)\right]
\;,
\\
\Hdiag^{(m)}
=& \;
V  \sum_{\vx; \, \vvl:\, r[S(\vx; \vvl)] \leq \rmax }  \left[
\,
\hd^{(m)}(\vx; \vvl_m) +
\bd^{(m)}(\vx; \vvl_m)\right]
\;,
\end{split}
\ee
where $\mathcal{S}(\mathbf{x}; \vvl)$ is the support of operator $\hb^{(m)}(\mathbf{x}; \vvl)$, and 
$r[\mathcal{S}(\mathbf{x}; \vvl)]$ is the largest Euclidean or Manhattan distance between any two plaquettes in the support $\mathcal{S}(\mathbf{x}; \vvl)$.  
Restrictions are imposed on the sum for simplicity, and they are sufficient but not necessary conditions for the model to satisfy criteria (i) and (ii).
Specifically, the sum is summed over all $\vx$ and $\vvl$ such that (a) each height flip operator is separated from other height flip operators by at least two plaquette distance, and (b)  the range $r[\mathcal{S}(\mathbf{x}; \vvl)]$  is smaller than a certain maximum range  $\rmax$.
Alternatively, one can define  $A[\mathcal{S}(\mathbf{x}; \vvl)]$ as the minimal size of a connected region of plaquettes that contain the support of operator $\hb^{(m)}(\mathbf{x}; \vvl)$, and define $\Amax$ as the maximum allowed support size.
The sum can then be defined such that condition (a) above is satisfied, and that (b) the minimal connected region that includes the support  $S(\mathbf{x}; \vvl)$ of $\hb(\mathbf{x}; \vvl)$ is no larger than $\Smax$.

It is worthwhile to comment that height flip operator defined in \eqref{eq:heigh_flip1} and \eqref{eq:heigh_flip2} corresponds to a loop update operation on a loop surrounding a \textit{single} plaquette. [Recall that given a loop of alternating link with and without dimers. A loop update is defined as flipping occupied links to unoccupied links in the loop , and vice versa.] One can define a variations of hQDM using height flip operator $b^{(-1)(\vx)}$ with support larger than a single plaquette, corresponding to loop updates operation on a larger loop.

As an example, the hQDM for $m=0$ can then be written as 
\be
\begin{split}
\Hres^{(0)}
=&
-t 
  \sum_{\vx; \vvl:\, r[S(\vx; \vvl)] \leq \rmax } 
 \left[
\,
\hr^{(0) } (\vx; (\vl))
+
\br^{(0)} (\vx; (\vl)) \right]
=
-t 
 \sum_{\vx;\, \vvl:\, r[S(\vx; \vvl)] \leq \rmax }  \left[
\,
\br_\vx \hr_{\vx+\vl}
+
\hr_\vx \br_{\vx+\vl}
\right]
\;,
\\
\Hdiag^{(0)}
=& \;
V \sum_{\vx;\, \vvl:\, r[S(\vx; \vvl)] \leq \rmax } \left[
\,
\hd^{(0) } (\vx; (\vl))
+
\bd^{(0) } (\vx; (\vl)) \right]
= 
V  \sum_{\vx; \vvl:\, r[S(\vx; \vvl)] \leq \rmax } \left[
\,
\bd_\vx \hd_{\vx+\vl}
+
\hd_\vx \bd_{\vx+\vl}
\right]
\;.
\end{split}
\ee
For $m=0$ and $\rmax=3$, the model given in Eq.~\eqref{eq:hQDM_def} is recovered (see for \ref{app:higher_m} for details).
%
%


\subsection{A particular short-range hQDM with general $m$
} \label{app:higher_m}

\begin{figure}[H]
\centering
\includegraphics[width=0.45 \textwidth  ]{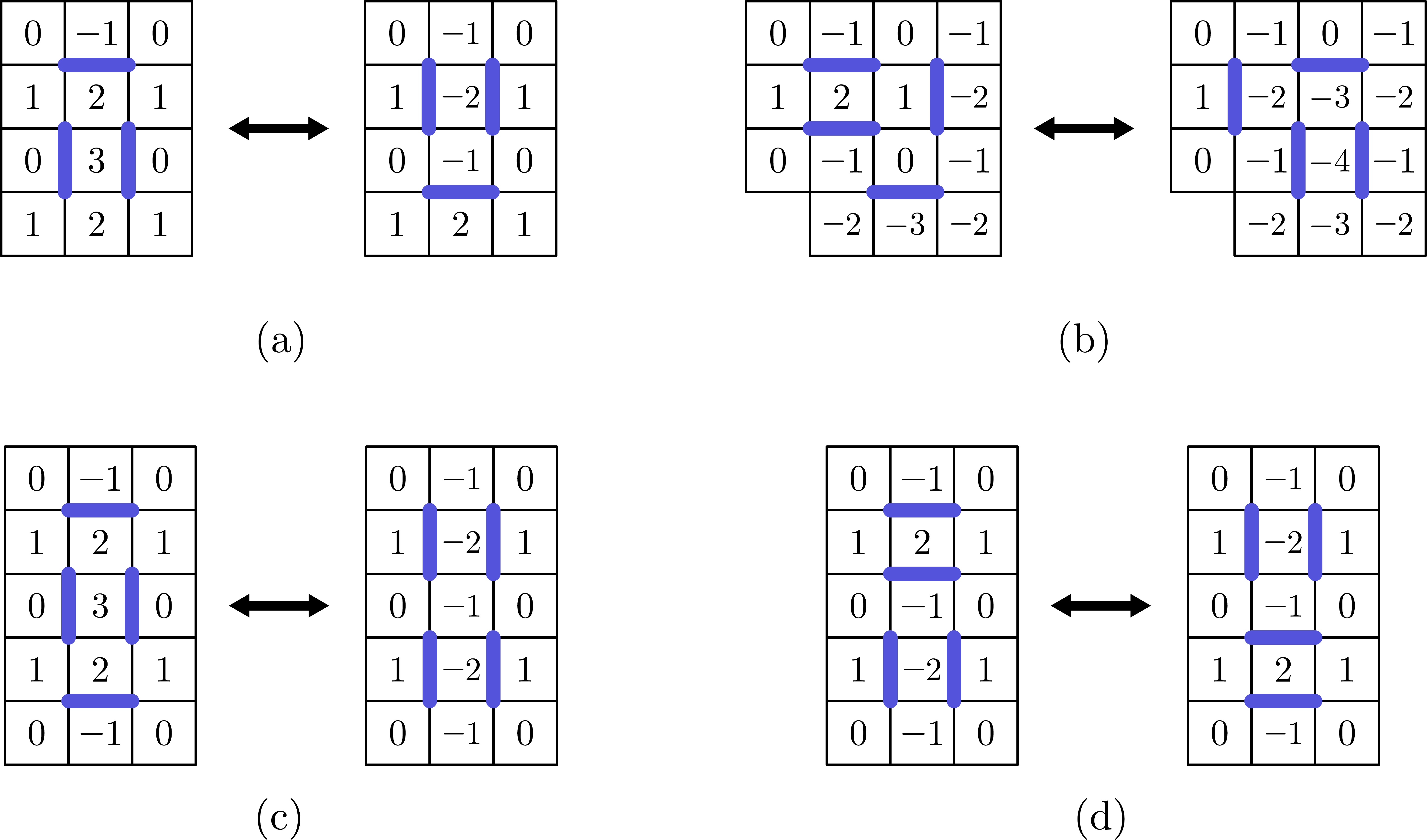}
    \caption{\label{fig:hQDM_derivation}
    (a) The only dimer move that involve two plaquettes does not preserve total height.  
    (b-d) The only dimer transformations that involve three adjacent plaquettes. (b) and (c) cannot conserve heights because the transformation can be written as a loop update of a loop surrounding three plaquettes -- such loop updates always only increase or decrease uniformly the heights within the loop.
    (d) The only dimer transformation on three adjacent plaquettes which preserve the total height. 
    }
\end{figure}

Here we define a particular short-range hQDM for general $m$, explicitly state it for $m=0, 1$ and $2$, and show that this short-range hQDM is in fact minimal-range hQDM for $m=0$.
The short-range hQDM is defined as 
\be\label{eq:hQDM_gen_def_short_range}
\begin{split}
H^{(m)} = & \, \Hres^{(m)} + \Hdiag^{(m)} \;,
\\
\Hres^{(m)}
=&
-
\frac{t}{2^{m \, \mathrm{mod} \, 2}}
 \sum_{\vx; \,i=1,2 } \left[
\,
\hr^{(m)}(\vx; \vvl_m( i)) +
\br^{(m)}(\vx; \vvl_m( i))\right]
\;,
\\
\Hdiag^{(m)}
=& \;
\frac{V}{2^{m \, \mathrm{mod} \, 2}}  \sum_{\vx; \, i=1,2 }  \left[
\,
\hd^{(m)}(\vx; \vvl_m( i)) +
\bd^{(m)}(\vx; \vvl_m( i))\right]
\;,
\end{split} 
\ee
with 
\be
\vvl_m( i) = (\vl_{0}(i), \vl_{1}(i), \dots, \vl_{m}(i)) =  (  2 \evec_i , 2 \evec_{\tilde{i}}, 4\evec_i, 4\evec_{\tilde{i}},\dots  )
= (2^{\lfloor a/2 \rfloor + 1} \evec_{((i + (a \,\mathrm{mod}\, 2)) \, \mathrm{mod}\,  2 )  +1}  )_{a=0,1,\dots,m}
\ee
where $i=1,2$ labels the horizontal and vertical direction, and $\tilde{i}=2,1$ if $i=1,2$ respectively.  $\lfloor \cdot \rfloor$ is the floor function.
The factor of 2 in Eq.~\eqref{eq:hQDM_gen_def_short_range} is to remove overcounting of terms for $m$, where the summands with $i=1$ and $2$ coincide. 
For $m=0$, we reproduce the Hamiltonian in Eq.~\eqref{eq:hQDM_def} reproduced below as 
\be \label{eq:hQDM_def2}
\begin{split}
\Hres^{(0)}
=&
-t \sum
\left(
\,
\left|
\vcenter{\hbox{  \includegraphics[scale = 0.1]{2020_Quantum_Dimer_1_color_v1}
}}
\right\rangle
\left\langle
\vcenter{\hbox{  \includegraphics[scale = 0.1]{2020_Quantum_Dimer_2_color_v1}
}}
\right|
+
\left|
\vcenter{\hbox{  \includegraphics[scale = 0.1]{2020_Quantum_Dimer_2_color_v1}
}}
\right\rangle
\left\langle
\vcenter{\hbox{  \includegraphics[scale = 0.1]{2020_Quantum_Dimer_1_color_v1}
}}
\right|
\right)
\\
\Hdiag^{(0)}
=&
\, 
V \, \sum \left(
\,
\left|
\figeq{0.1}{2020_Quantum_Dimer_1_color_v1}
\right\rangle
\left\langle
\vcenter{\hbox{  \includegraphics[scale = 0.1]{2020_Quantum_Dimer_1_color_v1}
}}
\right|
+
\left|
\vcenter{\hbox{  \includegraphics[scale = 0.1]{2020_Quantum_Dimer_2_color_v1}
}}
\right\rangle
\left\langle
\vcenter{\hbox{  \includegraphics[scale = 0.1]{2020_Quantum_Dimer_2_color_v1}
}}
\right|
\right)
\end{split}\;,
\ee
where again the sum sums over all possible vertically- or horizontally-aligned next-to-nearest-neighbour plaquette pairs.
Furthermore, for $m=0$, we show that this short-range hQDM is indeed the minimal-range hQDM, by enumerating all possible resonant and diagonal term with plaquette distance two in  in Fig.~\ref{fig:hQDM_derivation}. 
The hQDM on a square lattice for $m=1$ and $m=2$ can be defined respectively with
\be \label{eq:hQDM_quad1}
\begin{aligned} 
H_{\mathrm{res}}^{(1)}
=&
-t \sum
\left(
\,
\left|
\vcenter{\hbox{  \includegraphics[scale = 0.1]{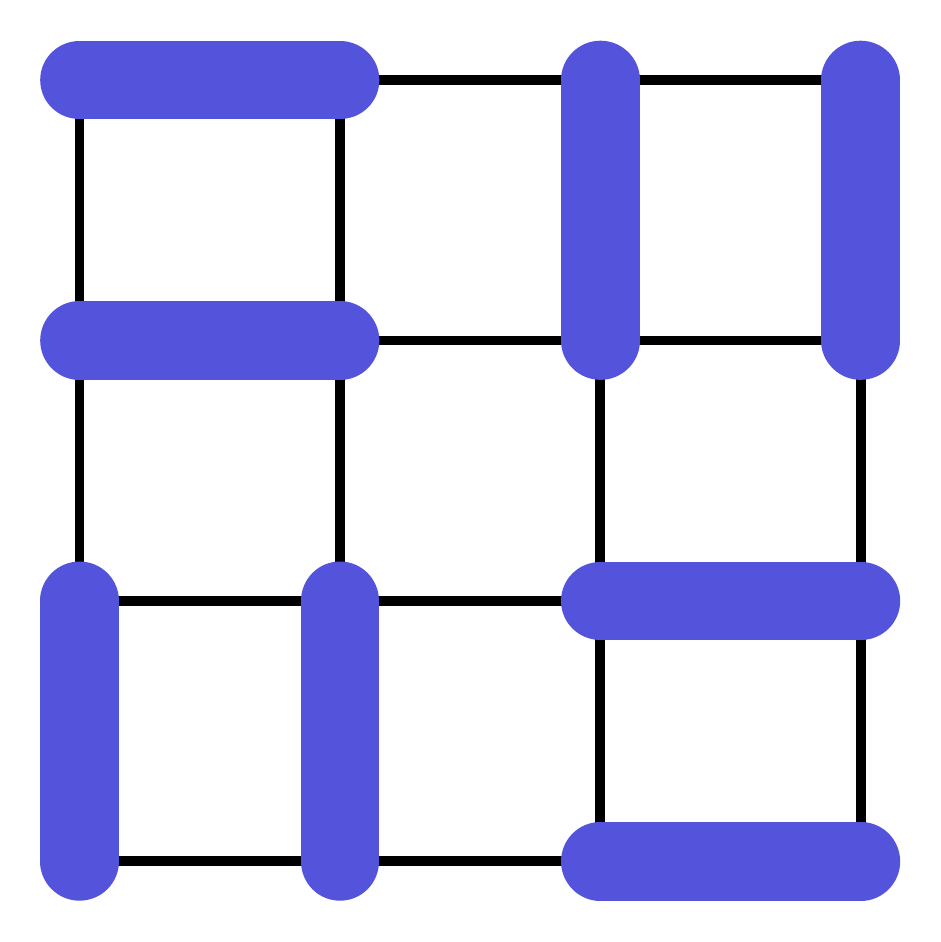}
}}
\right\rangle
\left\langle
\vcenter{\hbox{  \includegraphics[scale = 0.1]{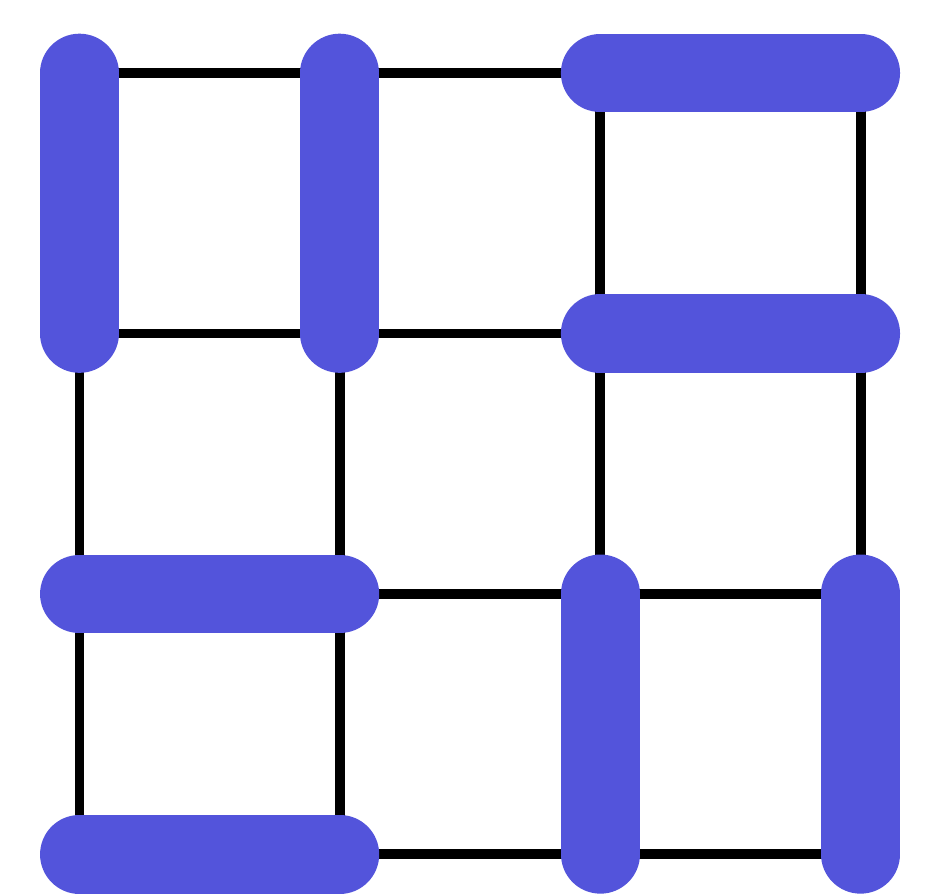}
}}
\right|
\right.
\left.
+
\left|
\vcenter{\hbox{  \includegraphics[scale = 0.1]{2020_Quantum_Dimer_3rd_moment_2_color_v1}
}}
\right\rangle
\left\langle
\vcenter{\hbox{  \includegraphics[scale = 0.1]{2020_Quantum_Dimer_3rd_moment_1_color_v1}
}}
\right|
\,
\right)  
\\
H_{\mathrm{diag}}^{(1)}
= &
V \sum \left(
\,
\left|
\figeq{0.1}{2020_Quantum_Dimer_3rd_moment_1_color_v1}
\right\rangle
\left\langle
\vcenter{\hbox{  \includegraphics[scale = 0.1]{2020_Quantum_Dimer_3rd_moment_1_color_v1}
}}
\right|
\right.
\left.
+
\left|
\vcenter{\hbox{  \includegraphics[scale = 0.1]{2020_Quantum_Dimer_3rd_moment_2_color_v1}
}}
\right\rangle
\left\langle
\vcenter{\hbox{  \includegraphics[scale = 0.1]{2020_Quantum_Dimer_3rd_moment_2_color_v1}
}}
\right|
\,
\right)  
\end{aligned} \;, 
\ee
and 
\be  \label{eq:hQDM_m2_model}
\begin{aligned}
\Hres^{(2)}
=&
\, t \, \sum \left(
\,
\left|
\figeq{0.1}{2020_Quantum_Dimer_4th_moment_1}
\right\rangle
\left\langle
\vcenter{\hbox{  \includegraphics[scale = 0.1]{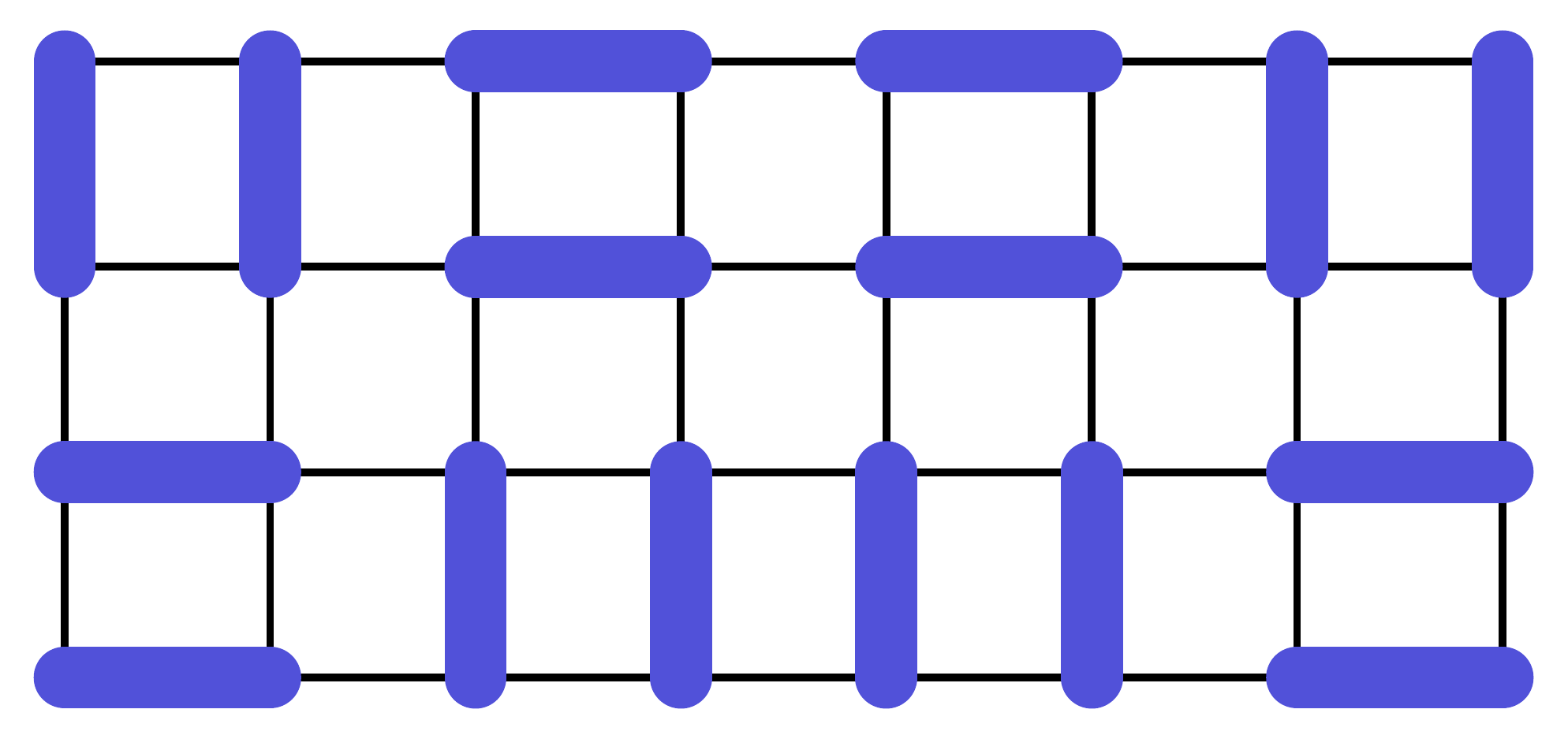}
}}
\right|
\right.
\left.
+
\left|
\vcenter{\hbox{  \includegraphics[scale = 0.1]{2020_Quantum_Dimer_4th_moment_2}
}}
\right\rangle
\left\langle
\vcenter{\hbox{  \includegraphics[scale = 0.1]{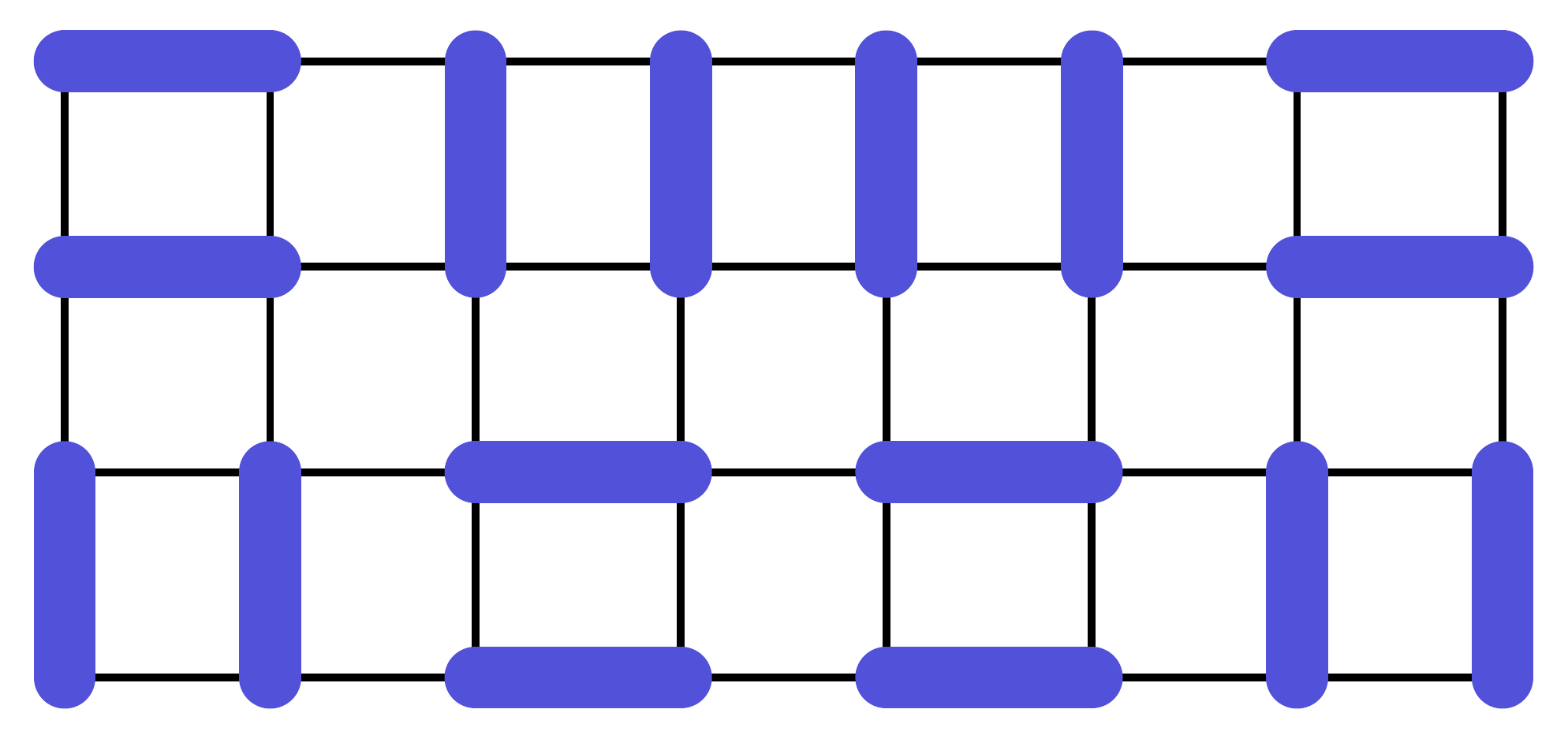}
}}
\right|
\,
\right)  
\\
\Hdiag^{(2)}
=&
\, V \, \sum \left(
\,
\left|
\figeq{0.1}{2020_Quantum_Dimer_4th_moment_1}
\right\rangle
\left\langle
\vcenter{\hbox{  \includegraphics[scale = 0.1]{2020_Quantum_Dimer_4th_moment_1}
}}
\right|
\right.
\left.
+
\left|
\vcenter{\hbox{  \includegraphics[scale = 0.1]{2020_Quantum_Dimer_4th_moment_2}
}}
\right\rangle
\left\langle
\vcenter{\hbox{  \includegraphics[scale = 0.1]{2020_Quantum_Dimer_4th_moment_2}
}}
\right|
\,
\right) 
\end{aligned}
\;.
\ee
where the sum is summing over all possible positions where one can act the operator on.
The higher-multipole-conserving of this particular short-range hQDM  can be straightforwardly generalized. 
Lastly, note that for $m=2$, Eq.~\eqref{eq:hQDM_m2_model} is not the only candidate model for \textit{minimal-range} hQDM. The following model also preserve the $m$-th multipole for $m=2$, 
\be  \label{eq:hQDM_3by3_hop}
\begin{aligned}
\Hres^{(2)}
=&
\, t \, \sum \left(
\,
\left|
\figeq{0.1}{2020_Quantum_Dimer_m2_alt_c1}
\right\rangle
\left\langle
\vcenter{\hbox{  \includegraphics[scale = 0.1]{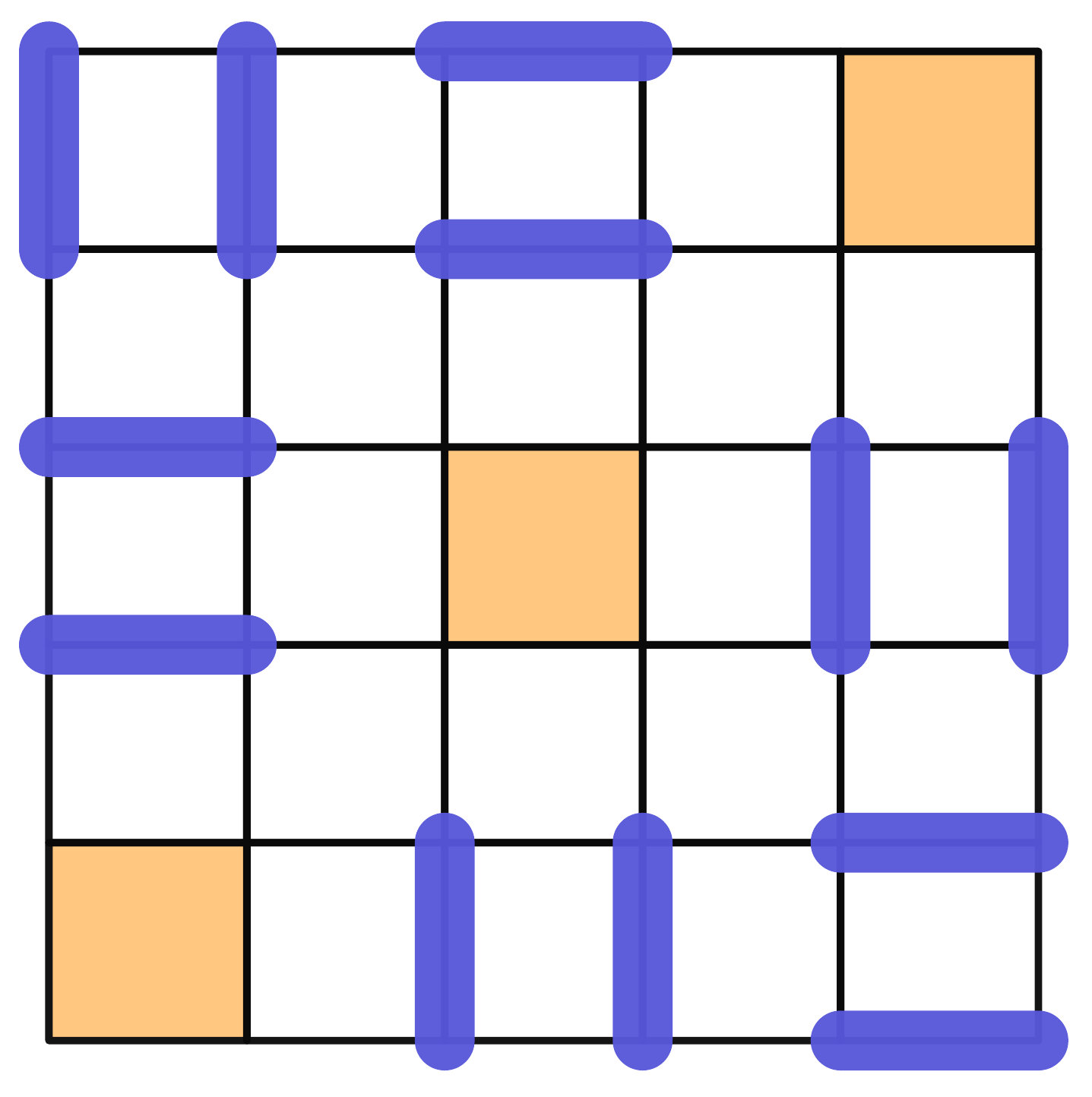}
}}
\right|
\right.
\left.
+
\left|
\vcenter{\hbox{  \includegraphics[scale = 0.1]{2020_Quantum_Dimer_m2_alt_c2}
}}
\right\rangle
\left\langle
\vcenter{\hbox{  \includegraphics[scale = 0.1]{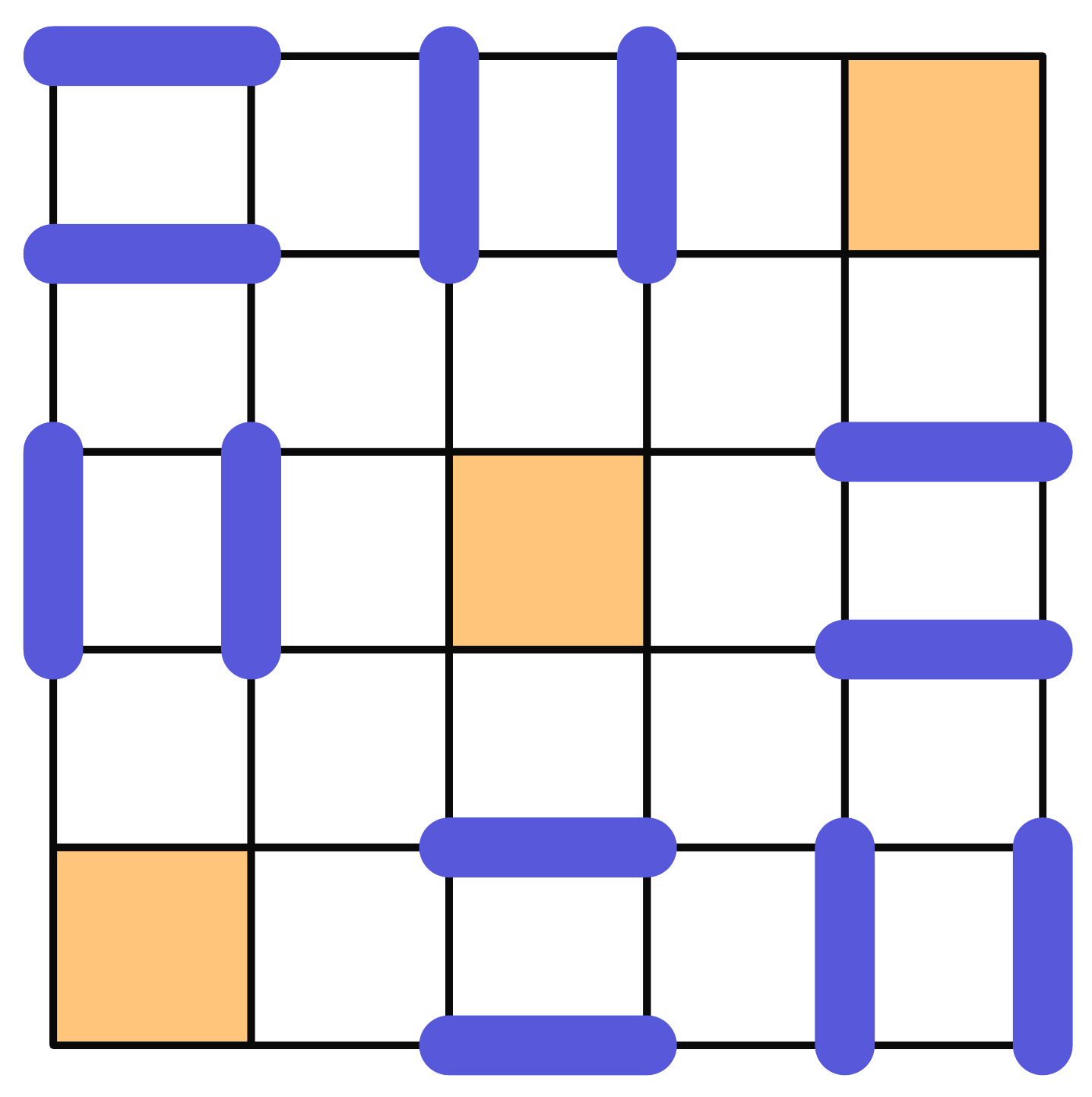}
}}
\right|
\,
\right)  \;.
\\
\Hdiag^{(2)}
=&
\, V \, \sum \left(
\,
\left|
\figeq{0.1}{2020_Quantum_Dimer_m2_alt_c1}
\right\rangle
\left\langle
\vcenter{\hbox{  \includegraphics[scale = 0.1]{2020_Quantum_Dimer_m2_alt_c1}
}}
\right|
\right.
\left.
+
\left|
\vcenter{\hbox{  \includegraphics[scale = 0.1]{2020_Quantum_Dimer_m2_alt_c2}
}}
\right\rangle
\left\langle
\vcenter{\hbox{  \includegraphics[scale = 0.1]{2020_Quantum_Dimer_m2_alt_c2}
}}
\right|
\,
\right) 
\end{aligned}
\;,
\ee
where again the sum is summing over all possible positions, and the orange plaquettes represent any dimer configurations. 
These two models, Eqs.~\eqref{eq:hQDM_m2_model} and \eqref{eq:hQDM_3by3_hop}, can both be considered minimal range model if the range of the model is parametrized by the largest Manhattan distance between any two plaquettes in the support of a term in the Hamiltonian. The model in Eq.~\eqref{eq:hQDM_3by3_hop} will have a shorter interaction range than Eq.~\eqref{eq:hQDM_m2_model},  if the interaction range is parametrized by the Euclidean distance instead.

\subsection{An example of longer-range hQDM with $m=0$}\label{app:hQDM_longer}
Here we explicitly write a longer-range (but still short-range) hQDM in Eq.~\eqref{eq:hQDM_gen_def} with $m=0$ and a maximum range of $\rmax = \sqrt{5}$ in Euclidean plaquette distance: %
\be
H^{(0)}_{\rmax= \sqrt{5}} =  
\Hres^{(0)}(r= 2)
+
\Hres^{(0)}(r= \sqrt{5})
+
\Hdiag^{(0)}(r= 2)
+
\Hdiag^{(0)}(r= \sqrt{5}) \;, 
\ee
where $\Hres^{(0)}(r= 2)$ and $\Hdiag^{(0)}(r= 2)$ are given in Eq.~\eqref{eq:hQDM_def2} and where
\be \label{eq:hQDM_knight}
\begin{split}
\Hres^{(0)}(r= \sqrt{5})
=&
-t \sum
\left(
\,
\left|
\vcenter{\hbox{  \includegraphics[scale = 0.1]{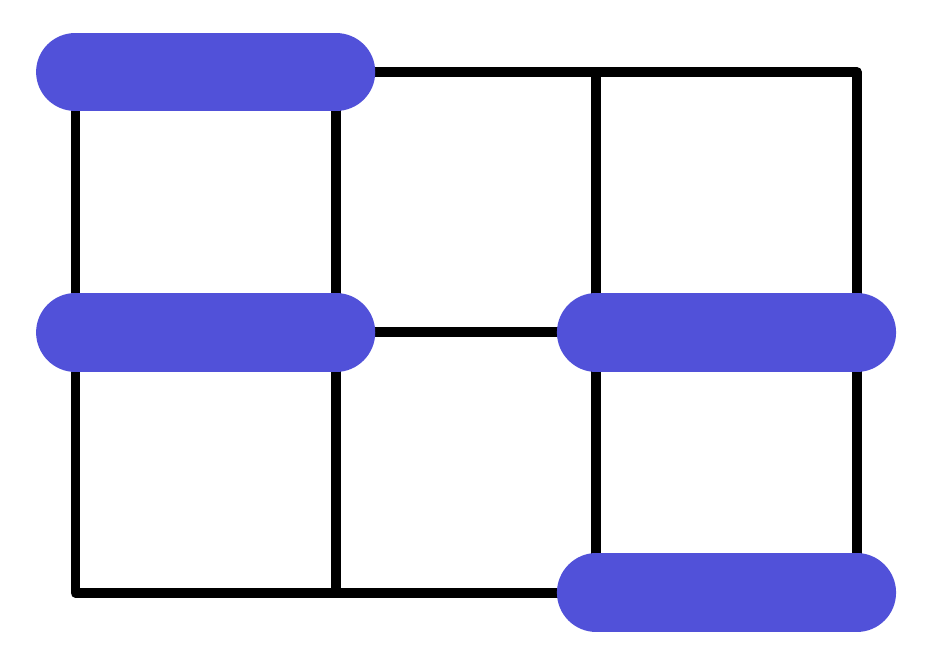}
}}
\right\rangle
\left\langle
\vcenter{\hbox{  \includegraphics[scale = 0.1]{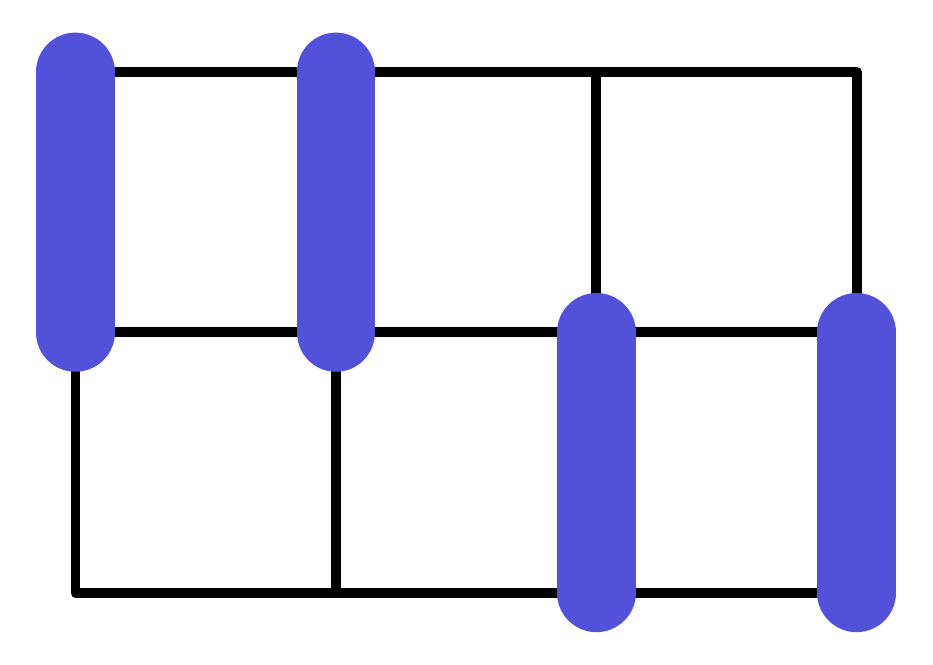}
}}
\right|
+
\left|
\vcenter{\hbox{  \includegraphics[scale = 0.1]{2022_Quantum_Dimer_knight2}
}}
\right\rangle
\left\langle
\vcenter{\hbox{  \includegraphics[scale = 0.1]{2022_Quantum_Dimer_knight1}
}}
\right|
\right)
\\
\Hdiag^{(0)}(r= \sqrt{5})
=&
\, 
V \, \sum \left(
\,
\left|
\figeq{0.1}{2022_Quantum_Dimer_knight1}
\right\rangle
\left\langle
\vcenter{\hbox{  \includegraphics[scale = 0.1]{2022_Quantum_Dimer_knight1}
}}
\right|
+
\left|
\vcenter{\hbox{  \includegraphics[scale = 0.1]{2022_Quantum_Dimer_knight2}
}}
\right\rangle
\left\langle
\vcenter{\hbox{  \includegraphics[scale = 0.1]{2022_Quantum_Dimer_knight2}
}}
\right|
\right)
\end{split}\;,
\ee
where the sum is over all possible knight moves on the square lattice (including knight moves of the opposite orientation which we omit for brevity).
%


\section{Analytically tractable limits in the phase diagram}
\subsection{Ground states as $V/t \to \infty$}

\begin{figure}[H]
\centering
\includegraphics[width=0.9 \textwidth  ]{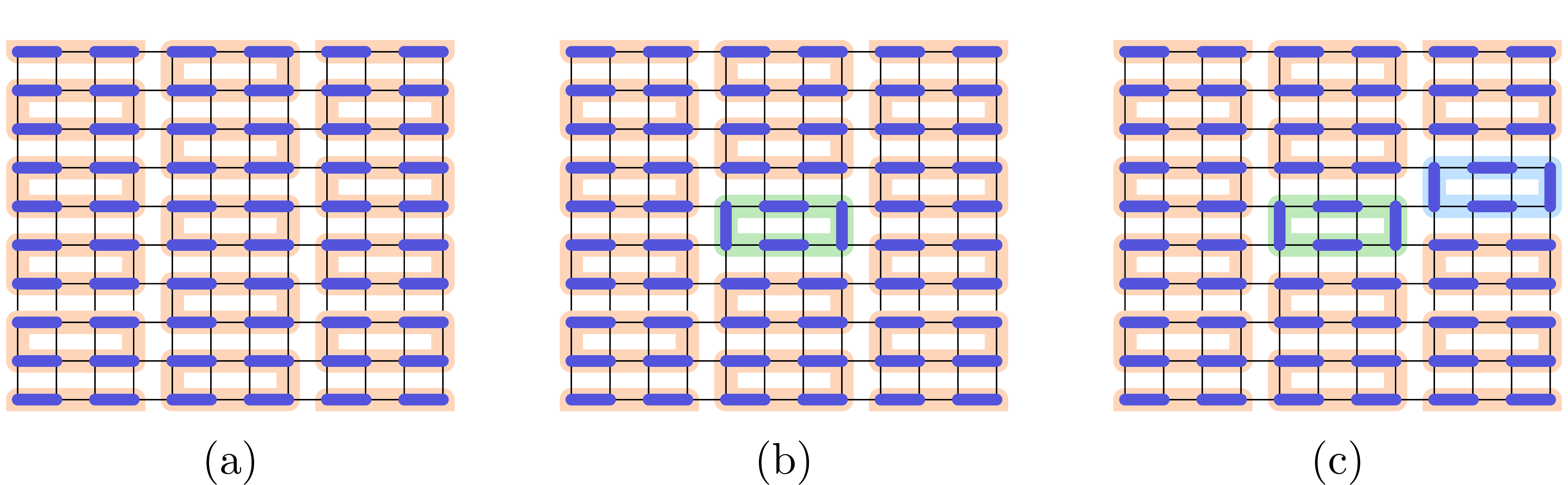}
    \caption{
    (a) The columnar state which has
    the maximal number of horizontal parallel dimer pairs.
    The orange lines indicate loops of links where loop updates can be performed to obtain a new state. A loop update turns occupied links to unoccupied links in the loop, and vice versa.
    (b) Along the green line, a loop update has been performed. No flippable plaquette pairs have been created.
    (c) Along the green and blue lines, loop updates has been performed. No flippable plaquette pairs have been created.
    } \label{fig:staggered_ground_state_extensive}
\end{figure}

\begin{proposition}\label{prop:gs_vt_infty}
The height-conserving qantum dimer model for general $m$ on the square lattice defined in Eq.~\eqref{eq:hQDM_gen_def_short_range} has at least $O\left(\alpha^{L^2} \right)$ degenerate ground states in the limit $V/t \to \infty$ for some constant $\alpha>1$. These ground states have zero energy, and are inactive, i.e. they belong to disconnected sectors in the Hilbert space.
\end{proposition}

\begin{proof}
Consider the columnar state with the maximal number of horizontal parallel dimer pairs, as shown in Fig.~\ref{fig:staggered_ground_state_extensive}a, which contains no flippable plaquette pairs, i.e. it is a ground state in the limit of  $V/t \to \infty$.
Partition the configuration with strings (orange lines in Fig.~\ref{fig:staggered_ground_state_extensive}a).
One has the freedom to shift dimers within a string while creating no flippable vertical dimers, i.e. the new state is also a ground state.
To prove this, it suffices to see that no vertical dimer pairs are created within the rectangle where a string transformation has been made (green lines in Fig.~\ref{fig:staggered_ground_state_extensive}b), and that no vertical dimer pairs are created between vertical bonds after two separate string transformations have been made (green and light blue lines in Fig.~\ref{fig:staggered_ground_state_extensive}c).
This implies that there is at least  $O\left(\alpha^{ L^2} \right)$  ground states for some $\alpha >1$.
\end{proof}

\subsection{Ground states as $V/t \to - \infty$}\label{app:foursubsector}

\begin{figure}[H]
\centering
\includegraphics[width=0.9 \textwidth  ]{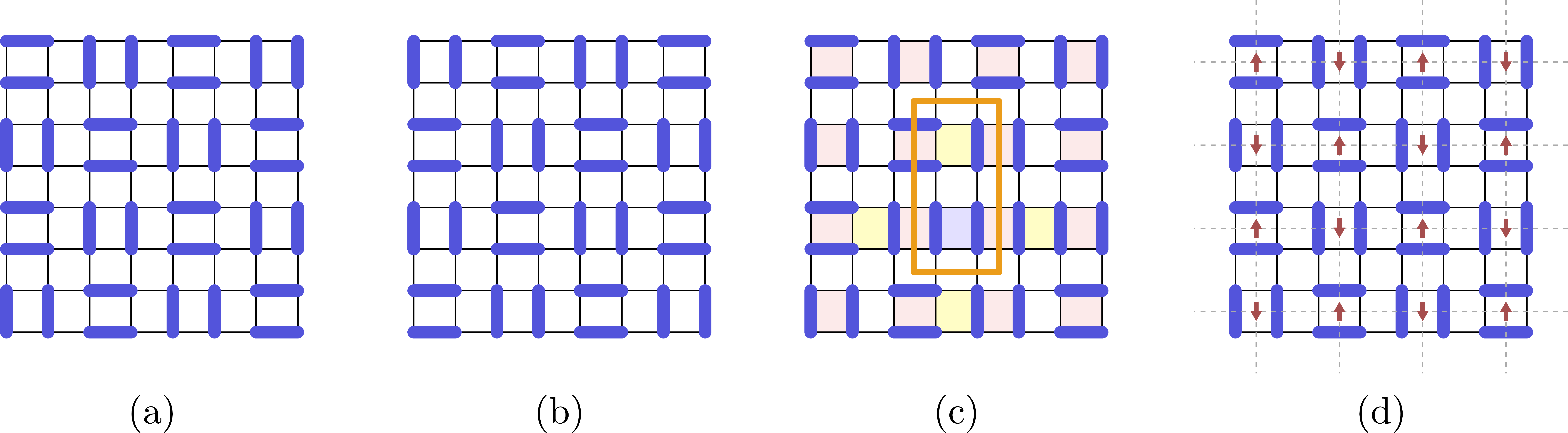}
    \caption{
    (a-b) Two ground states in the $V/t \to - \infty$ limit.
    These two ground states have the maximal number of flippable dimer pairs, which live in the same sublattices.
    (c) Red plaquettes are where the flippable dimer pairs live. The blue and yellow plaquettes can host flippable plaquettes, but not flippable plaquette pairs.
    (d) Mapping from the ground state in (a) to an antiferromagnet ground state in the 2D XXZ Model.
    } \label{fig:grid_ground_state_proof}
\end{figure}

\begin{proposition}
The height-conserving qantum dimer model for general $m$ on the square lattice defined in Eq.~\eqref{eq:hQDM_gen_def_short_range} has eight degenerate ground states in the limit $V/t \to -\infty$. The eight ground states belong to four separate disconnected sectors [whose physics is described by an emergent constrained spin models given in given in Proposition \ref{prop:mapping}].
\end{proposition}

\begin{proof}
Suppose that linear system size  $L$ is a multiple of 4 and the periodic boundary condition is imposed.
We first note that there are maximally $  L^2 /2$ flippable plaquette pairs on a given square lattice. %
To find the 8 ground states, without loss of generality, we place a parallel pair of dimers on the top left plaquette of the square lattice.
We then place the dimer pairs on next-to-nearest-neighbour plaquettes such that these plaquettes form flippable plaquette pairs.
We repeat the above procedure except we change the initial plaquette position and parallel dimer pairs orientation.
This allow us to identify eight and only eight distinct states that maximize the number of flippable plaquette pairs.
Two of the  eight ground states are shown in Fig.~\ref{fig:grid_ground_state_proof} a and b.

We now show that the two degenerate ground states shown in  Fig.~\ref{fig:grid_ground_state_proof}a and \ref{fig:grid_ground_state_proof}b belong to the same subsector in the Hilbert space (since the state  Fig.~\ref{fig:grid_ground_state_proof}a  can be evolved under the time evolution of $H$ into  Fig.~\ref{fig:grid_ground_state_proof}b),
but they are not connected to the other six ground states (not shown) which are related to the first two ground states by shifting the entire pattern one plaquette to the right, bottom, bottom right.
To show this, we note that under the time evolution governed by the hQDM, ground states shown in
Fig.~\ref{fig:grid_ground_state_proof}a and \ref{fig:grid_ground_state_proof}b can only have pair plaquette flips on the red plaquettes in Fig.~\ref{fig:grid_ground_state_proof}c.
Suppose the ground state has undergone pair plauqette flips on red plauqettes only, and it is only possible to have a parallel dimer pair located directly in between two next-to-nearest-neighbor red plaquettes.
Without loss of generality, suppose this parallel dimer pair is located at the blue plaquette in Fig.~\ref{fig:grid_ground_state_proof}c.
To have a pair plaquette flip outside of the red plquettes, we must perform a pair plaquette flip on the blue plaquettes and one of the yellow plaquettes, e.g. the colored plaquettes in the orange loop.
In order for this plaquette flip to be possible, one of the yellow plaquettes must have a parallel dimer configuration of the opposite orientation from the parallel dimer pair at the blue plaquette.
This is impossible since we have only flipped red plaquettes. Therefore, the groundstate in Fig.~\ref{fig:grid_ground_state_proof}a is only connected to the one in Fig.~\ref{fig:grid_ground_state_proof}b.
This argument can be analogously applied to the other pairs of ground states, and general values of $m$.
\end{proof}


\section{Structure of Hilbert space}\label{app:subsector}
\subsection{Domain walls of staggered dimers (DWSD) as blockades}\label{app:dwsd}
\begin{figure}[H]
\centering
\includegraphics[width=0.7\textwidth  ]{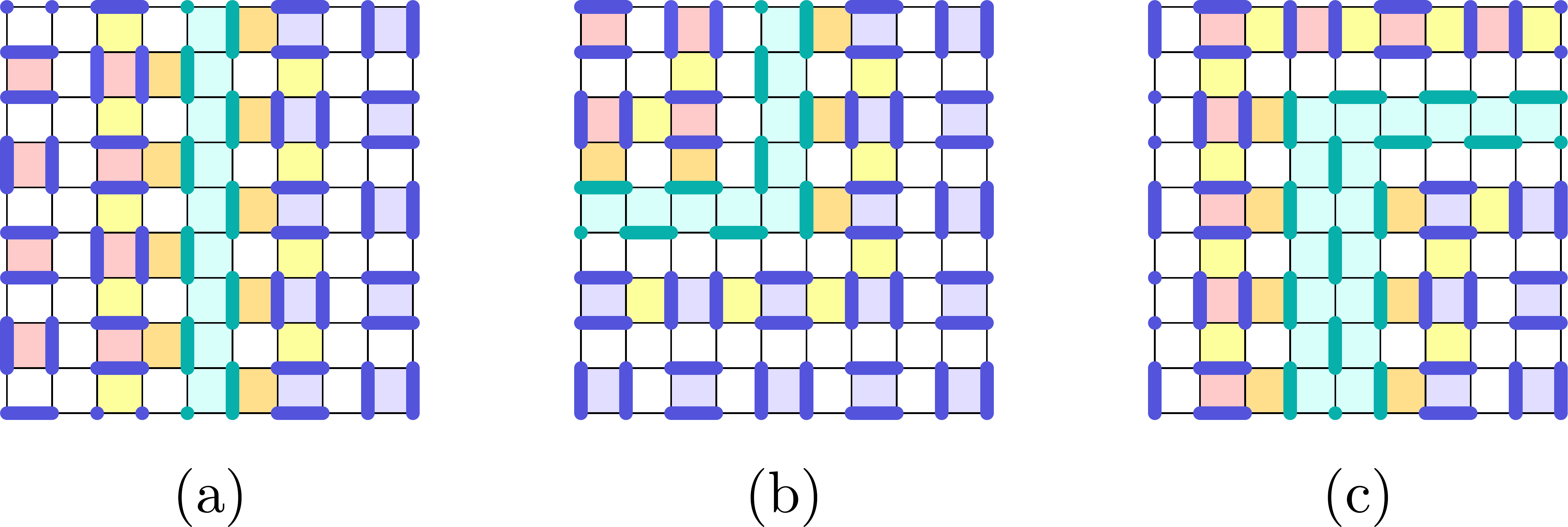}
    \caption{
    (a) A domain wall of stagered dimers (DWSD) in cyan, with flippable plaquette pairs only on sublattice A on the left of DWSD, and on B on the right o DWSD.
    (b) DWSD with a turn.
    (c) DWSD with a turn and with various thicknesses.
    } \label{fig:dwsd}
\end{figure}

Here we will define a region of dimers which can be used as ``blockades'' such that no terms in the Hamiltonian of hQDM can act non-trivially in that region. These concepts will be used in other sections of Appendix \ref{app:subsector} to construct Hilbert space subsectors.
\begin{definition}
A domain wall of staggered dimers (DWSD) is a connected region of plaquettes, where no plaquettes are occupied by parallel dimer pairs.  
\end{definition}
Examples of DWSD is given in Fig.~\ref{fig:dwsd} in cyan. 
\begin{proposition}
Consider the The height-conserving qantum dimer model for general $m$ on the square lattice defined in Eq.~\eqref{eq:hQDM_gen_def_short_range}.
Suppose a DWSD, $W$, separates two regions of plaquettes $R_1$ and $R_2$, such that $R_1$ and $R_2$ contain flippable plaquette pairs in only one of the four sublattices $\{A, B, C, D\}$. 
No terms in the Hamiltonian of hQDM can act non-trivially on regions of plaquettes in both $W$ and $R_1$, and on regions of plaquettes in both $W$ and $R_2$. In other words, DWSD acts as a blockade between $R_1$ and $R_2$.
\end{proposition}
\begin{proof}
By construction, $R_1$ only contains flippable plaquette pairs in only one sublattice, say sublattice A. 
On the edge of region $R_1$, parallel plaquette pairs can be formed on sublattice A (e.g. red plaquettes in Fig.~\ref{fig:dwsd}), on plaquettes in between sublattice A along the edge (e.g. yellow plaquettes in Fig.~\ref{fig:dwsd}), on plaquettes in between sublattice A and DWSD (e.g. orange plaquettes in Fig.~\ref{fig:dwsd}). 
By construction of the DWSD, there are no parallel plaquettes pairs that can form with plaquettes in yellow or orange.
Therefore, there are no parallel plaquettes in DWSD, and therefore no terms in the Hamiltonian of hQDM can act non-trivially on the plaquettes in both region $R_1$ and $W$. 
Similar arguments apply to $R_2$ and $W$. 
Examples are given in Fig.~\ref{fig:dwsd}.
\end{proof}

\subsection{  Exponential number of Krylov subspace of size $1$ }
\begin{proposition}
The height-conserving qantum dimer model for general $m$ on the square lattice defined in Eq.~\eqref{eq:hQDM_gen_def_short_range} has $O\left(\alpha^{ L^2}\right)$ number of Krylov subspace of size $1$ for some constant $\alpha>1$.
\end{proposition}
\begin{proof}
A frozen or inactive state is a state that is annihilated by all terms in the Hamiltonian. 
Such a state is also a ground state in the limit of $V/t\to \infty$, which has no flippable plaquette pairs. 
Hence the claim of  $O\left(\alpha^{ L^2}\right)$ number of Hilbert space  subsectors of size $1$ for $\alpha >1$ is equivalent to proposition \ref{prop:gs_vt_infty}.
\end{proof}

\subsection{ Exponential number of Krylov subspace of size $O(1)$}

\begin{figure}[H]
\centering
\includegraphics[width=0.9 \textwidth  ]{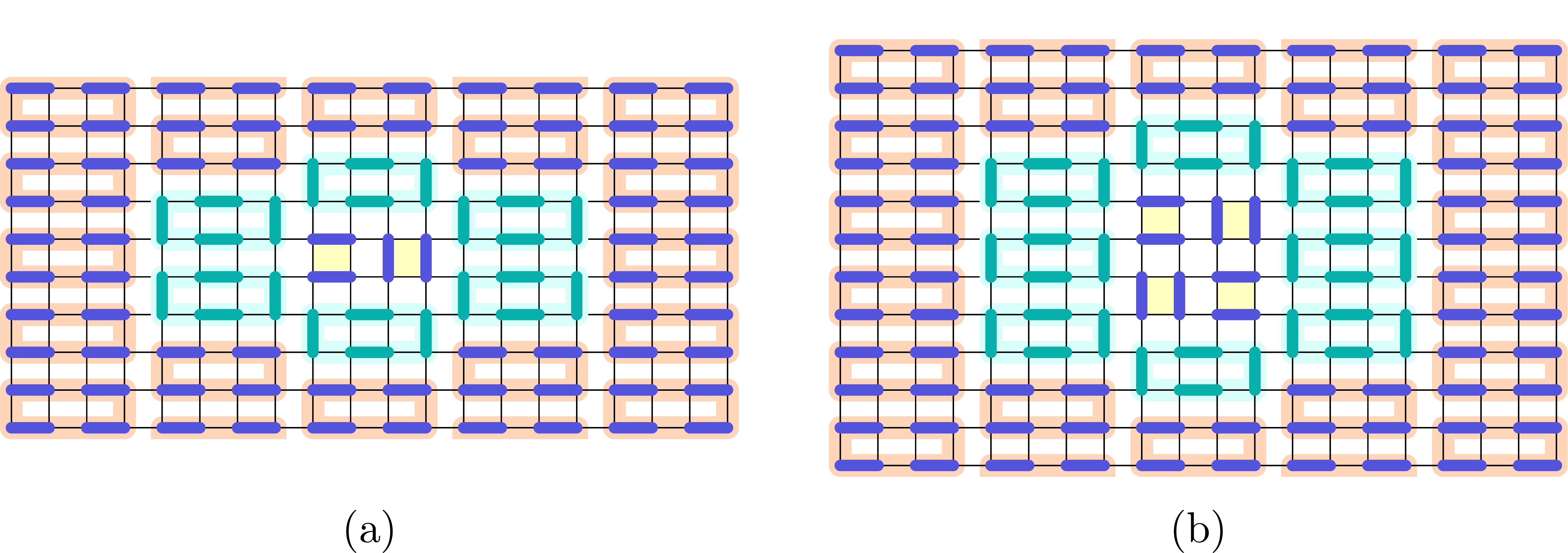}
    \caption{
    (a) $O(\exp L^2)$ number of subsectors of two configurations for $m=0$. 
    The Hamiltonian can act non-trivially to the two plaquettes in yellow.
    The cyan regions are constructed to ensure the active region is localized.
    Other subsectors of two configurations can be generated by performing loop updates on orange lines outside of the active region (see \ref{fig:staggered_ground_state_extensive}).
    (b) $O(\exp L^2)$ number of subsectors of size two for $m=1$.
    } \label{fig:1dim_sbusector}
\end{figure}

\begin{proposition}
The height-conserving qantum dimer model for general $m$ on the square lattice defined in Eq.~\eqref{eq:hQDM_gen_def_short_range} has $O\left(\alpha^{ L^2}\right)$ number of Krylov subspace of size $2$ for some constant $\alpha>1$.
\end{proposition}

\begin{proof}
We explicitly construct $O\left(\alpha^{ L^2}\right)$ number of Hilbert space subsectors of size $2$ for  hQDM in Eq.~\eqref{eq:hQDM_gen_def_short_range} with general $m$ and $\alpha>1$.
Consider first the case of $m=0$. We start with Fig.~\ref{fig:staggered_ground_state_extensive}a, and call each region of three plaquettes enclosed by an orange loop a ``brick''.
We replace a brick of three plaquettes with a flippable plaquette pair as in Fig.~\ref{fig:1dim_sbusector}a (green plaquettes). 
We arrange all bricks around the first brick with dimer configuration given in cyan in  Fig.~\ref{fig:1dim_sbusector}a. 
This configuration belongs to a Hilbert space subsector of dimension two, since there is a single active but isolated brick. 
We generate $O\left(\alpha^{ L^2}\right)$ number of such subsectors by performing loop updates (recall that a loop update turns occupied links to unoccupied links in the loop, and vice versa.) to orange loops in Fig.~\ref{fig:1dim_sbusector}a. 

For $m=1$, we can apply the same procedure to construct a Hilbert space subsector of size 2 in Fig.~\ref{fig:1dim_sbusector}b. Again, $O\left(\alpha^{ L^2}\right)$ number of such subsectors can be generated by performing loop updates orange loops. The cases for higher $m$ can be straightforwardly generalized.
\end{proof}

\subsection{Exponential number  of Krylov subspace of size $O(\alpha^{L^2})$ with $\alpha>1$}

\begin{proposition}
The height-conserving qantum dimer model for general $m$ on the square lattice defined in Eq.~\eqref{eq:hQDM_gen_def_short_range} has  $O(\gamma^{L^2})$ number of Hilbert space subsectors with size $O(\alpha^{L^2})$ with some constants $\alpha,\gamma > 1$.
\end{proposition}

\begin{proof}
We explicitly construct $O(\gamma^{L^2})$ number of Hilbert space subsectors of size $O(\alpha^{L^2})$ for  hQDM Eq.~\eqref{eq:hQDM_gen_def_short_range} with general $m$.
We construct configurations using blocks of dimer configurations of size $\ell$-by-$\ell$.
To be concrete, we take $\ell = 12$  illustrated in Fig.~\ref{fig:subsector_puzzle}, 
but $\ell$ can be arbitrary as long as they are sufficiently large, e.g. in the main text we used $\ell=8$ in Fig.~\ref{fig:subsector_puz}.
Blocks A, B, C, and D are chosen to have bulks which contain flippable plaquette pairs (that preserve the $m$-th multipole moments) in sublattices A, B, C, and D respectively. 
Block F is an example of frozen block of size 12-by-12.
These blocks have boundaries of DWSD (see Appendix \ref{app:dwsd}), such that they 
are frozen under the actions of pairwise plaquettes flips along the boundary, or between the boundary and the bulk.
Next we piece these blocks together. 
We  claim that no flippable plaquette pairs are created between two parallel edges of blocks. For example, there should not be flippable plaquette pairs between blocks A and C in the top left of Fig.~\ref{fig:subsector_puzzle}b. 
This is true, since two aligned blocks will only create dimer pairs along the horizontal or vertical direction only, but not dimer pairs in both the horizontal and vertical directions.

Furthermore, we show that no flippable plaquette pairs are created at the corner in between four blocks that are pieced together. For example, there should not be  flippable plaquette pairs at the corner between between blocks B, D, D, A in the bottom left of Fig.~\ref{fig:subsector_puzzle}b. 
This is true since there exists no dimers across the thick black lines in Fig.~\ref{fig:subsector_puzzle}b by construction, and hence plaqutte dimerpairs of opposite orientations cannot be found adjacent to each other. 
Alternatively, one can enumerate all possible combinations of edges and corners between blocks, and show that no flippable plaquette pairs are created.

This construction allows us to have $O(\delta^{L^2/\ell^2})$ number of Hilbert space subsectors, since each choice of combination of blocks, e.g. the one given in Fig.~\ref{fig:subsector_puzzle}b, is disconnected to another distinct combination of blocks.
Furthermore, each Hilbert space subsector size is of order exponential in the system size of the system, since each block (suppose we exclude block F) can be constructed to have a lower bound dimension of $f(\ell)$ for some function $f$. The size of the Hilbert space subsector is larger than $O\left(f(\ell)^{ L^2/\ell^2}\right)$. 
To connect with the proposition above, we have $\gamma=  \delta^{1/\ell^2}$ and $\alpha \geq f(\ell)^{ 1/\ell^2}$.

\end{proof}

\begin{figure*}[t!]
\centering
\includegraphics[width=0.9 \textwidth  ]{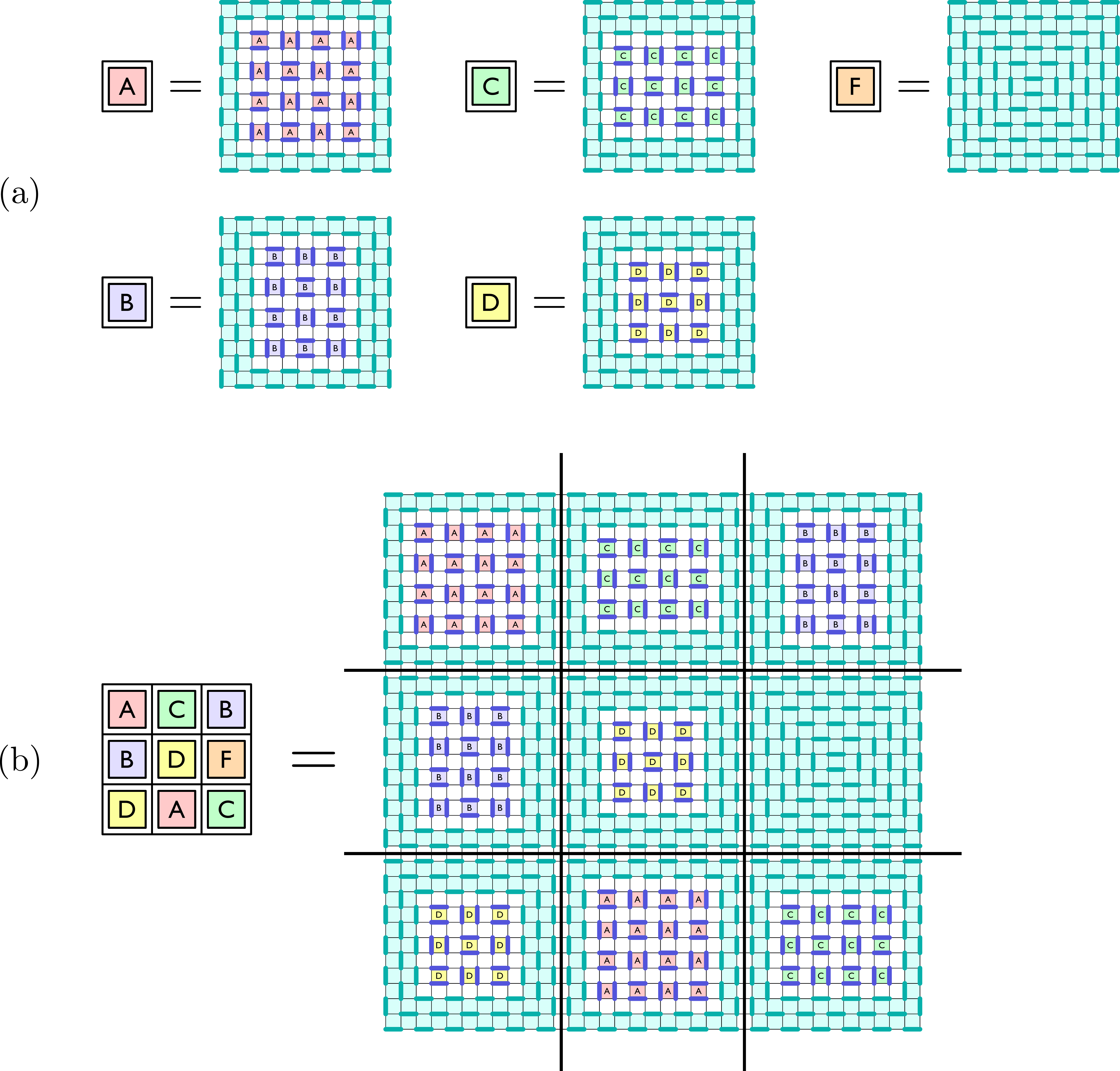}
    \caption{
   (a) Example 12-by-12 blocks with active regions of flippable dimers in sublattices A, B, C, and D, surrounded by inactive domain wall of staggered dimers (DWSD). 
   Block F denotes a block that is completely frozen or inactive. 
   The size of these blocks can be modified.
   (b) Exponential number of subsectors constructed by piecing together blocks given in (a).  
    } \label{fig:subsector_puzzle}
\end{figure*}

\section{Mapping between hQDM and spin models}
\begin{proposition} \label{prop:mapping}
Consider the short range and $m$-th-multipole-conserving hQDM on the square lattice defined in Eq.~\eqref{eq:hQDM_gen_def_short_range}.
Consider the Krylov subspace of state $\ket{\psi}_X$ where all plaquettes in sublattice $X \in \{\text{A, B, C, D} \}$ (see Fig.~\ref{fig:phasediag} bottom left) are occupied with parallel dimer pairs. 
The Hamiltonian restricted to subsector $X$ can be mapped to a constrained spin model using the mapping $g$ 
\be \label{eq:mapping}
\begin{aligned}
\; 
    \left|
\hspace{-0.05cm}
\vcenter{\hbox{  \includegraphics[scale = 0.1]{2020_05_Quantum_Dimer_Color_plaquette_2}
}}
\hspace{-0.05cm}
\right\rangle_\vx
\quad &
\xrightarrow{\; \; g \; \; } 
\quad 
    \left|
\downarrow
\right\rangle_{\vxt}
\\
\;
    \left|
\hspace{-0.05cm}
\vcenter{\hbox{  \includegraphics[scale = 0.1]{2020_05_Quantum_Dimer_Color_plaquette_1}
}}
\hspace{-0.05cm}
\right\rangle_\vx
\quad &
\xrightarrow{\; \; g \; \; } 
\quad 
    \left|
\uparrow
\right\rangle_{\vxt}
\end{aligned}
\;\; .
\ee
where the coordinate $\{\vxt \}$ labelling the plaquettes in  sublattice $X$ given by
\be \label{eq:coord}
\tilde{x} = \left( 
\frac{x_1 + \eta_1}{2},
\frac{x_2 + \eta_2}{2}
\right) \;,
\ee
with 
$(\eta_1, \eta_2)= (1,1), (0,1),(1,0),(0,0)$ for case A, B, C, and D respectively.
Hamiltonian now acts on the Krylov subspace of state $\ket{\psi}_X$ as 
\be
H^{(m)}_{\mathrm{hQDM}} 
\big|_{
\mk
\left(
{\ket{\psi}_{X}}
\right)
, \{\vx\} 
}
= H^{(m)}_{\mathrm{spin}}  \big|_{
\mk
\left(
g\left( \ket{\psi}_X \right)
\right)
, \{\vxt\} 
}
\ee
where $H^{(m)}_{\mathrm{spin}}$ is given by \eqref{eq:hQDM_gen_def} using the mapping in  \eqref{eq:mapping} and \eqref{eq:coord}.

\end{proposition}
\vspace{0.2cm}

\begin{proof}
This is a direct consequence of the fact given in Sec. \ref{app:foursubsector} that in the said Hilbert space subsector (in which all plaquettes in sublattice $X \in \{\text{A, B, C, D} \}$ are occupied with parallel dimer pairs), only terms supported in the sublattice $X$ act non-trivially (see Fig.~\ref{fig:grid_ground_state_proof}).   
\end{proof}

As an example, for $m=0$, $H^{(m=0)}_{\mathrm{hQDM}}$ in Eq. \eqref{eq:hQDM_def} can be mapped to
\be
\left.
H^{(m=0)}_{\mathrm{hQDM}} \right|_{\mathrm{X}, \{\vx \}}
=
\left.
H_{\mathrm{spin}}^{(m=0)}
\right|_{\{\vxt \} } 
= 
H_{\mathrm{XXZ}}
\Big|_{ \{ \vxt \} } 
\ee

\be
H_\mathrm{spin}^{(0)} = H_{\mathrm{XXZ}} = -t \sum_{\left\langle \vxt, \vxt' \right\rangle} \left( 
\splus_{\vxt} \sminus_{\vxt'}  + \sminus_{\vxt} \splus_{\vxt'}  
\right) 
+
\frac{V}{2} \sum_{\left\langle \vxt, \vxt' \right\rangle} \left(1 -  \sz_{\vxt} \sz_{\vxt'}  \right) 
\ee
where the sum is over all nearest neighboring sites on a square lattice.


The sum sums over all possible vertically- or horizontally-aligned next-to-nearest-neighbour plaquette pairs.

\section{Dimer pair correlation functions and structure factors}
We define the pair dimer operator on a plaquette at position $\xvec$ as 
\begin{equation}
\begin{aligned}
    \dop_{\hordimers,\xvec}&=
      \left|
\hspace{-0.05cm}
\vcenter{\hbox{  \includegraphics[scale = 0.1]{2020_05_Quantum_Dimer_Color_plaquette_1}
}}
\hspace{-0.05cm}
\right\rangle_\xvec
\left\langle
\hspace{-0.05cm}
\vcenter{\hbox{  \includegraphics[scale = 0.1]{2020_05_Quantum_Dimer_Color_plaquette_1}
}}
\hspace{-0.05cm}
\right|_\xvec
    \\
    \dop_{\verdimers,\vx}&=
      \left|
\hspace{-0.05cm}
\vcenter{\hbox{  \includegraphics[scale = 0.1]{2020_05_Quantum_Dimer_Color_plaquette_2}
}}
\hspace{-0.05cm}
\right\rangle_\xvec
\left\langle
\hspace{-0.05cm}
\vcenter{\hbox{  \includegraphics[scale = 0.1]{2020_05_Quantum_Dimer_Color_plaquette_2}
}}
\hspace{-0.05cm}
\right|_\xvec \;. 
\end{aligned}
\end{equation}
and the pair correlation function $C_{ij,\xvec-\xvec'}$ between $\dop_{i,\xvec}$ and $\dop_{j,\xvec'}$ as:
\begin{equation}\label{eq:cor_func}
    C_{ij,\xvec-\xvec'}=\frac{\langle \dop_{i,\xvec}\dop_{j,\xvec'}\rangle-\langle \dop_{i,\xvec}\rangle\langle \dop_{j,\xvec'}\rangle}{
    \langle \dop_{i,\xvec}\dop_{j,\xvec}\rangle-\langle \dop_{i,\xvec}\rangle \langle\dop_{j,\xvec}\rangle}
\end{equation}
where $i,j=\hordimers,\verdimers$ and $\xvec-\xvec'$ is the position difference. 
\begin{figure}[htp]
\centering
\includegraphics[width=0.85 \textwidth]{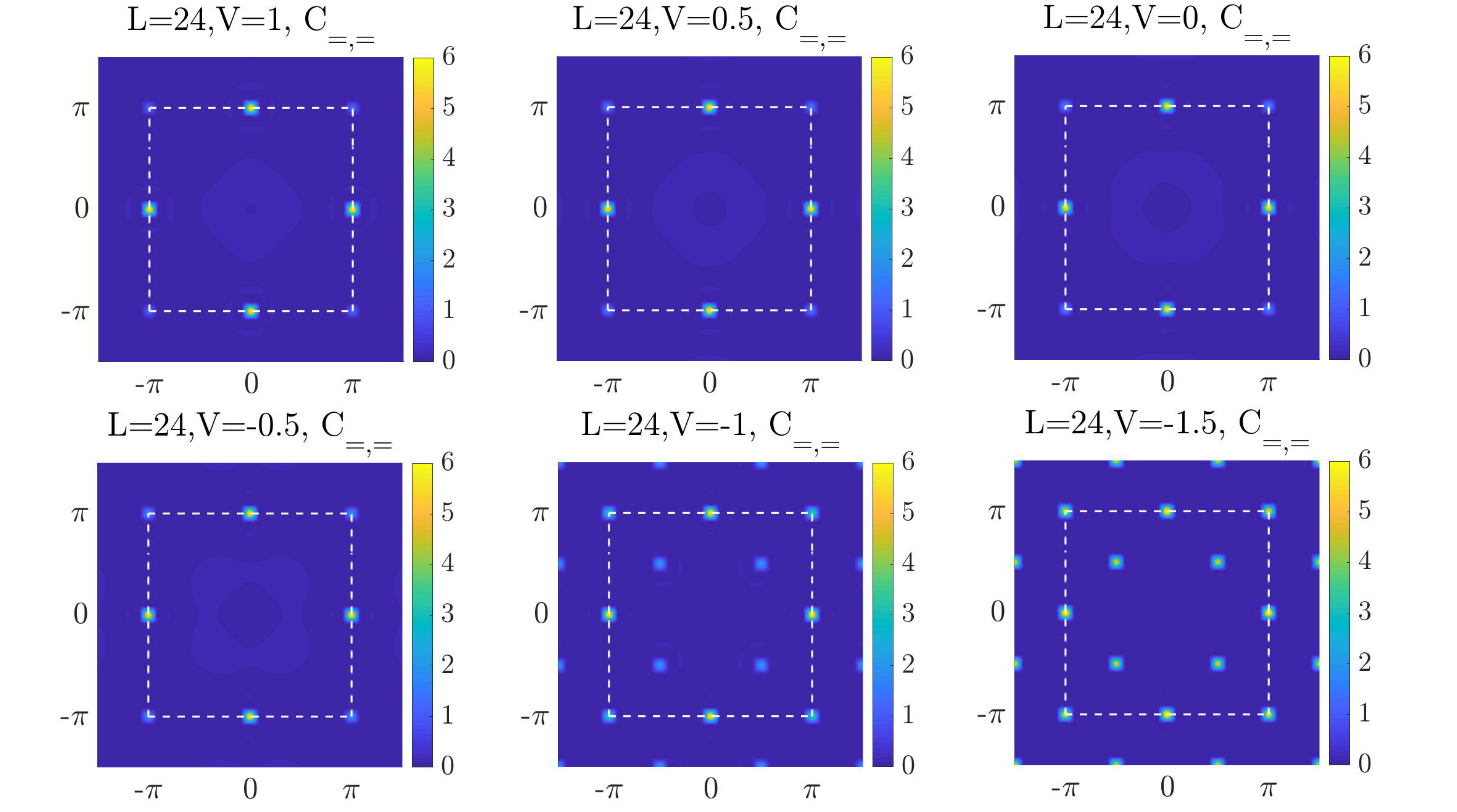}
    \caption{The structure factor of dimer pair correlation function $C_{ij,\xvec-\xvec'}$ for $V/t $ from $1$ to $-1.5$ in steps of $0.5$.
    The linear lattice size is $L=24$, and the temperature $T=1/\beta=1/24$.
    }
    \label{SK}
\end{figure}
\begin{figure}[htp]
\centering
\includegraphics[width=0.85 \textwidth  ]{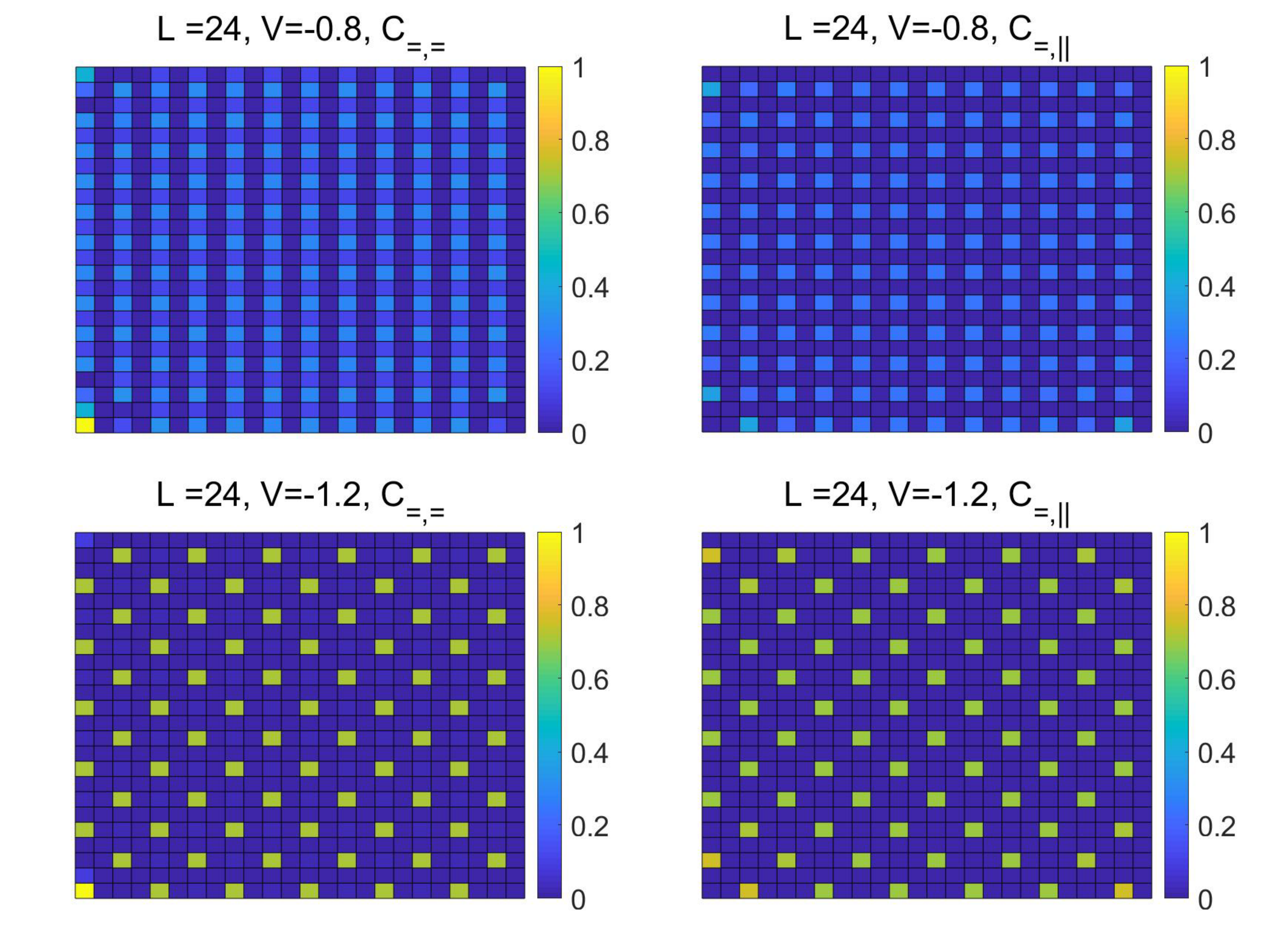}
    \caption{The correlation functions $C_{ij,\vx-\vx'}$ in the real space with $V/t = -0.8$ (upper panels) and $V=-1.2$ (lower panels) and $\vx'$ fixed at the bottom left plaquette. 
    The linear system size is $L=24$, and the temperature $T=1/\beta=1/24$. 
    The color coding denotes the probably of finding a dimer pair at a given plaquette, with yellow means high probability and blue means low probability, the other colors are in between. 
    }
    \label{SR}
\end{figure}

In the momentum space, and in the VBS phase, the structure factor of $C_{\hordimers,\hordimers}$ should have sharp peaks at $\mathrm{X}$ points, that is, $(\pm \pi/2, \pm\pi/2)$ points as shown in Fig.3 (d) of main text. %
In a plaquette phase, the $\mathrm{X}$ peaks of $C_{\hordimers,\hordimers}$ disappear but the peaks of $\mathrm{M}$, i.e., $ (\pm\pi, 0)$ and $(0, \pm\pi)$ points.
As shown in Fig.~\ref{SK}, the numerical results are consistent with the two phases. 

In the real space, we can also read the above structures of correlation functions in Fig.\ref{SR}. 
In the VBS phase (for example at $V/t=-1.2$), the correlation $C_{\hordimers,\hordimers}$ shows that the horizontal dimer pairs favor translation periodicity of four.
Similarly, the $C_{\hordimers,\verdimers}$ shows the vertical dimer pairs also favor the same translation period, but vertical dimers prefer to stay two plaquettes away from the horizontal ones, such that a flippable plaquette pairs are formed. 
In the plaquette phase (at $V=0$), indeed we see the translation period of both horizontal and vertical correlations becomes two. These features are consistent with the corresponding VBS and plaquette phases.
\begin{figure}[htp]
\centering
\includegraphics[width=0.5 \textwidth  ]{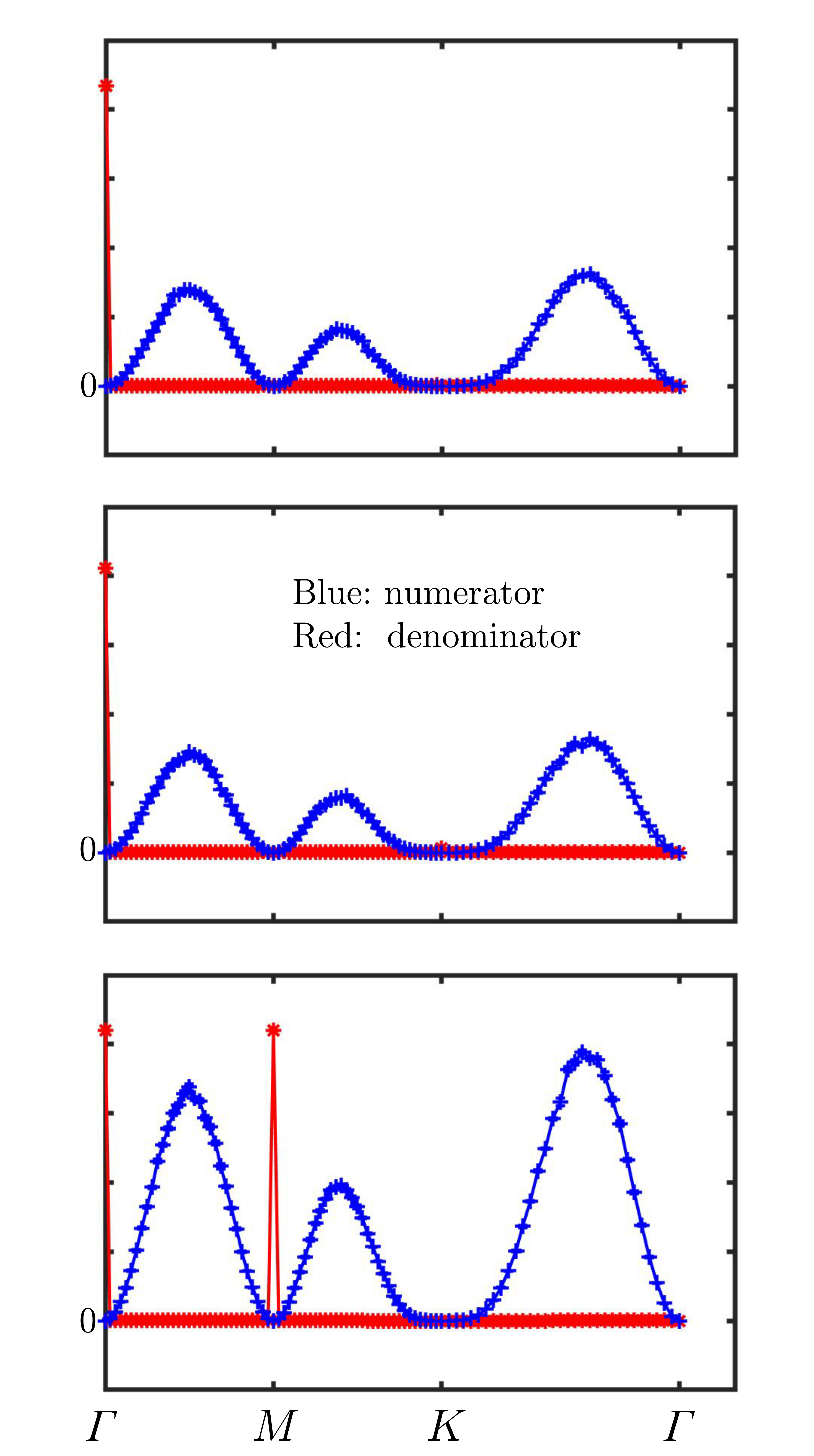}
    \caption{Three examples for the numerator and denominator parts of Eq.~\eqref{SMAEQ} at RK point. The denominator has some singular points but stand on the zero point of numerator. This ensures the effectiveness of SMA.
    }
    \label{ff}
\end{figure}

\section{Single Mode Approximation}
In this paper, we use single mode approximation~\cite{moessner2003three, L_uchli_2008} (SMA) upon the QMC data to demonstrate the large dynamical exponent $z$ of the gapless modes.
We define a dimer number operator $n_{\alpha}( \rvec )$ which has eigenvalues 0 or 1 upon acting on a state with or without a dimer along positive $\alpha$-direction at site $\rvec$, where $\alpha = x, y$.
The density operator in momentum space is given by
$\dop_\alpha( \Qvec)=\frac{1}{\sqrt{N}}
\sum_{\rvec}
\exp(-i \Qvec \cdot \rvec )
\,
n_\alpha (\rvec)$.
This allows us to define the state $|\Qvec \rangle= \dop_\alpha( \Qvec)\ket{\GS }$, where $\ket{\GS}$ is the ground state from QMC simulation.
We obtain a variational energy of
\be
\omega(\Qvec) \leq \omega_{\SMA}(\Qvec)
=
\frac{\langle \GS| [\dop_\alpha(-\Qvec),[H,\dop_\alpha(\Qvec)]]| \GS \rangle
}{
\langle \GS| \dop_\alpha(-\Qvec) \dop_\alpha(\Qvec)
|\GS \rangle} \;,
\label{SMAEQ}
\ee
which allows us to obtain an upper bound of $\omega(\Qvec) \leq \omega_{\SMA}(\Qvec)$ and hence information about the low-lying spectrum by studying the correlation of ground state.

SMA works well near gaplessness to reveal the dispersion of soft excitations. For example, while the density  $\dop_x(\Qvec)$ is a conserved quantity, i.e., $[H,\dop_x(\Qvec)]=0$, the energy $\omega(\Qvec)=\omega_{\SMA}(\Qvec)=0$ demonstrates the existence of gapless modes at $\Qvec$. As shown in Fig.\ref{ff}, The denominator has some singular points but stand on the zero point of numerator. This ensures the effectiveness of SMA.
The behaviours of $\omega_{\SMA}$ near $\Qvec$ is used to determine the bound of the soft mode dispersion in QDM~\cite{moessner2008review, L_uchli_2008}.

For hQDM, it is straightforward to show that there are at least four gapless modes at $\Qvec = (0,0), (\pi,0), (0,\pi)$ and $(\pi,\pi)$. Since the physics at $\Qvec = (\pi,0), (0,\pi)$ are related by symmetries, we compute in the main text the dispersion around the four momenta except $\Qvec = (0,\pi)$.

\end{document}